\def\draft{1}
\def\doubleblind{0}
\documentclass{article}

\usepackage{tcolorbox}
\usepackage{empheq}

\usepackage[margin=1in]{geometry}
\usepackage[utf8]{inputenc}
\usepackage{bm,bbm}
\usepackage{graphicx}
\usepackage{color,xcolor}
\usepackage{amsmath,amsfonts, amssymb}

\usepackage{amsthm,thmtools}

\PassOptionsToPackage{hyphens}{url}
\usepackage[colorlinks=true, allcolors=blue]{hyperref}
\usepackage[capitalise,nameinlink]{cleveref}

\usepackage[style=alphabetic, backend=biber, minalphanames=3, maxalphanames=4, maxbibnames=99]{biblatex}

\crefformat{section}{#2\S#1#3}
\crefformat{subsection}{#2\S#1#3}
\crefformat{subsubsection}{#2\S#1#3}

\declaretheoremstyle[bodyfont=\it,qed=\qedsymbol]{noproofstyle}

\numberwithin{equation}{section}

\declaretheorem[numberlike=equation]{observation}

\declaretheorem[name=Observation,numbered=no]{observation*}

\declaretheorem[numberlike=equation]{theorem}

\declaretheorem[name=Theorem,numbered=no]{theorem*}

\declaretheorem[numberlike=equation]{lemma}
\declaretheorem[name=Lemma,numbered=no]{lemma*}

\declaretheorem[name=Corollary,numbered=no]{corollary*}

\declaretheorem[numberlike=equation]{proposition}
\declaretheorem[name=Proposition,numbered=no]{proposition*}

\declaretheorem[numberlike=equation]{claim}
\declaretheorem[name=Claim,numbered=no]{claim*}

\declaretheorem[name=Conjecture,numbered=no]{conjecture*}

\declaretheorem[name=Question,numbered=no]{question*}

\declaretheoremstyle[bodyfont=\it]{defstyle} 

\declaretheorem[numberlike=equation,style=defstyle]{definition}
\declaretheorem[unnumbered,name=Definition,style=defstyle]{definition*}

\declaretheorem[unnumbered,name=Example,style=defstyle]{example*}

\declaretheorem[unnumbered,name=Notation=defstyle]{notation*}

\declaretheorem[unnumbered,name=Construction,style=defstyle]{construction*}

\declaretheoremstyle[]{rmkstyle} 

\newtheorem*{remark}{Remark}

\newcommand{\mkand}{\mathsf{Max}\text{-}k\mathsf{AND}}
\newcommand{\mtwoand}{\mathsf{Max}\text{-}2\mathsf{AND}}
\newcommand{\mthreeand}{\mathsf{Max}\text{-}3\mathsf{AND}}
\newcommand{\mdcut}{\mathsf{Max}\text{-}\mathsf{DICUT}}

\newcommand{\bias}{\mathsf{bias}}

\newcommand{\vecpos}{+\bm1}

\newcommand{\val}{\mathsf{val}}

\newcommand{\Snap}{\mathsf{Snap}_\Psi^\vect}

\newcommand{\obl}{\mathtt{Obl}^{\vect,\vecp}_k}
\newcommand{\Clauses}{\mathsf{Ptn}}
\newcommand{\PClauses}{\mathsf{PosPtn}}
\newcommand{\prob}{\mathsf{prob}}
\newcommand{\I}{\mathsf{Int}^\vect}
\newcommand{\ptn}{\mathsf{ptn}^\vect}
\newcommand{\flip}{\mathsf{flip}}

\newcommand{\aalg}{\alpha_{\mathrm{alg}}}
\newcommand{\alp}{\alpha_{\mathrm{LP}}}

\usepackage{singer-macros}

\title{Oblivious algorithms for the $\mkand$ problem}
\date{\today}
\ifnum\doubleblind=0\author{Noah G. Singer\thanks{Department of Computer Science, Carnegie Mellon University, Pittsburgh, PA, USA. Supported by an NSF Graduate Research Fellowship (Award DGE2140739). Email: \texttt{ngsinger@cs.cmu.edu}.}}\fi

\begin{document}

\maketitle

\begin{abstract}
    Motivated by recent works on streaming algorithms for constraint satisfaction problems (CSPs), we define and analyze \emph{oblivious algorithms} for the $\mkand$ problem. This generalizes the definition by Feige and Jozeph (Algorithmica '15) of oblivious algorithms for $\mdcut$, a special case of $\mtwoand$. Oblivious algorithms round each variable with probability depending only on a quantity called the variable's \emph{bias}.
    
    For each oblivious algorithm, we design a so-called \emph{factor-revealing} linear program (LP) which captures its worst-case instance, generalizing one of Feige and Jozeph for $\mdcut$. Then, departing from their work, we perform a fully explicit analysis of these (infinitely many!) LPs. In particular, we show that for all $k$, oblivious algorithms for $\mkand$ provably outperform a special subclass of algorithms we call ``superoblivious'' algorithms.
    
    Our result has implications for streaming algorithms: Generalizing the result for $\mdcut$ of Saxena, Singer, Sudan, and Velusamy (SODA'23), we prove that certain separation results hold between streaming models for \emph{infinitely many} CSPs: for \emph{every} $k$, $O(\log n)$-space sketching algorithms for $\mkand$ known to be optimal in $o(\sqrt n)$-space can be beaten in (a) $O(\log n)$-space under a random-ordering assumption, and (b) $O(n^{1-1/k} D^{1/k})$ space under a maximum-degree-$D$ assumption. Even in the previously-known case of $\mdcut$, our analytic proof gives a fuller, computer-free picture of these separation results.
\end{abstract}

\tableofcontents

\section{Introduction}

In this work, we study a restricted but natural class of randomized algorithms called \emph{oblivious algorithms} for a family of \emph{constraint satisfaction problems (CSPs)} called $\mkand$ for $k \geq 2$. In this problem, the algorithm is presented with a list of $m$ constraints on $n$ Boolean variables; each constraint specifies desired values for $k$ of the $n$ variables; and the goal is to satisfy the highest possible fraction of constraints.\footnote{Equivalently, each constraint is a disjunction of $k$ literals.} We begin by introducing these problems and algorithms informally and discussing the context and motivation for our work.

\subsection{Background and context}

$\mkand$ is the ``maximally expressive'' Boolean CSP: each constraint specifies \emph{exactly} what its $k$ variables must be assigned to. This makes it, in a sense, ``universal'' for $k$-ary Boolean CSPs. In particular, as observed by Trevisan~\cite{Tre98-alg}, an \emph{arbitrary} Boolean predicate $\phi$ of arity $k$  with $r$ satisfying assignments can be converted to $r$ ``disjoint'' applications of $\mkand$ constraints; this transformation makes an instance of $\textsf{Max-CSP}(\phi)$ into an instance of $\mkand$ and drops the value by a factor exactly $r$. In turn, this means that algorithms for $\mkand$ can approximate the acceptance probability of $k$-bit probabalistically checkable proofs (PCP) verifiers. As a consequence, numerous works have developed algorithms for $\mkand$ \cite{Tre98-alg,Has04,Has05,CMM09} as well as $\NP$-hardness-of-approximation results \cite{Tre98-hardness,ST98,ST00,EH08,ST09}; we now know that $\Theta(k/2^k)$-approximations are the best achievable in polynomial time assuming $\P \neq \NP$ \cite{CMM09,ST09}.

Further attention has been devoted to important special cases of $\mkand$. One particularly important example is the $\mdcut$ problem, a special case of $\mtwoand$ where each constraint is of the form ``$x \wedge \neg y$''. $\mdcut$ can be viewed alternatively as a directed graph optimization problem, where the goal is to find a directed cut $(S,T)$ maximizing the number of edges $(s,t)$ such that $s \in S$ and $t \in T$. Approximation algorithms for $\mdcut$ (and sometimes $\mtwoand$) were developed in \cite{GW95,FG95,MM01,LLZ02}, and its hardness-of-approximation was studied in \cite{Has01}. $\mthreeand$ was studied in \cite{Zwi98,TSSW00}.

The importance of $\mkand$, the extensive work on its polynomial-time approximability, and the ``expressiveness'' of its constraints has inspired significant study on its approximability in restricted algorithmic settings. For instance, Trevisan~\cite{Tre98-alg} showed that the natural linear programming (LP) relaxation for $\mkand$ beats the trivial (uniformly random rounding) algorithm's approximation ratio by a factor of $2$, and that this LP's ``nice'' structure allows it to be solved (approximately) by distributed algorithms. $\mdcut$ in particular has been studied extensively in various restrictive algorithmic frameworks and models, including ``combinatorial'' algorithms \cite{HZ01}, spectral partitioning algorithms \cite{ZDW+21}, local search algorithms \cite{Ali96,Ali97}, parallel algorithms \cite{BEF22}, near-linear time algorithms \cite{Ste10}, and online algorithms \cite{BL12}.

In the setting of streaming algorithms, there has been a steady flow of results in the past decade on $\mdcut$'s approximability in a variety of models \cite{GVV17,CGV20,SSSV23-random-ordering,SSSV23-dicut}. In particular, the recent works of Saxena, Singer, Sudan, and Velusamy~\cite{SSSV23-random-ordering,SSSV23-dicut} have demonstrated that $\mdcut$ exhibits a phenomenon previously unbeknownst to any CSP: It admits approximation algorithms in certain streaming regimes which beat the optimal algorithms in weaker regimes. Key to these works was an earlier investigation by Feige and Jozeph~\cite{FJ15} which defined and analyzed a class of simple algorithms, called ``oblivious'' algorithms, for $\mdcut$. It is the generalization of oblivious algorithms to $\mkand$ and their implications for streaming algorithms which are the focus of this paper.

In the (graph-theoretic view of the) $\mdcut$ problem, an oblivious algorithm is one which randomly rounds each vertex depending only its \emph{bias}, which is the (relative) difference between its in- and out-degrees. The natural generalization to $\mkand$ is an algorithm which randomly rounds each vertex depending only on the (relative) difference between its number of positive and negative appearances in the instance.\footnote{Or, more generally in weighted instances, the total weight of the clauses in which it appears positively vs. those in which it appears negatively.} These algorithms are simple enough that they can be implemented in various online or distributed settings, and this simplicity also powers their usefulness in the streaming context.

In another recent work, Boyland, Hwang, Prasad, Singer, and Velusamy~\cite{BHP+22} studied the streaming approximability of the $\mkand$ problem (along with other Boolean CSPs). They showed, by analyzing a framework of algorithms and lower bounds due to Chou, Golovnev, Sudan, and Velusamy~\cite{CGSV21-boolean}, that an optimal $o(\sqrt n)$-space sketching algorithm for $\mkand$ ``corresponds'', in a loose sense, to what we in this paper will call a \emph{superoblivious} algorithm. These are a subclass of oblivious algorithms which round using only \emph{which is more common}, positive vs. negative appearances, and not \emph{how much} more common. For the special case of $\mdcut$, for instance, the optimal algorithm achieves a ratio of $4/9$, and corresponds to a superoblivious algorithm which rounds a vertex with out-degree exceeding in-degree to the $S$-side of the cut with probability $2/3$. 

The key technical ingredient in the current work is the definition and analysis of a so-called \emph{factor-revealing linear program (LP)}. In such an LP, feasible solutions encode instances of the problem at hand (i.e., $\mkand$); when we fix an algorithm in the designated class (i.e., an oblivious algorithm), the objective function ``reveals'' the approximation ratio the algorithm achieves on any given instance. Similar programs were first studied in depth for facility location problems by Jain, Mahdian, Markakis, Saberi, and Vazirani~\cite{JMM+03}, and have been examined in other contexts such as online bipartite matching~\cite{MY11}. Our LP is a generalization of the one developed by Feige and Jozeph~\cite{FJ15} for $\mdcut$.

Feige and Jozeph~\cite{FJ15} showed, using computer analysis of their LP, that there is a $0.483$-approximation oblivious algorithm for $\mdcut$. The fact that $0.483 > 4/9$ --- that is, oblivious algorithms outperform $o(\sqrt n)$-space streaming algorithms for $\mdcut$ --- is precisely what was used in the works of \cite{SSSV23-random-ordering,SSSV23-dicut} to establish that there are improved approximations in stronger streaming models. However, this result of \cite{FJ15} was used as a black box. We believe that our approach, which contrasts between ``superoblivious'' and more general ``oblivious'' algorithms and attacks the corresponding LPs from an analytical perspective, gives a more natural and systematic explanation for why $\mdcut$ admits these improved approximations --- and implies that $\mkand$ does as well, for all $k$.

\subsection{Results}

Next, we turn to statements of our results. In our notation, an oblivious algorithm for $\mkand$ is denoted $\obl$, where $\vect$ and $\vecp$ are, respectively, a \emph{bias partition} which splits the space of possible biases (i.e., the interval $[-1,+1]$) into discrete intervals, and a \emph{rounding vector} $\vecp$ specifying a probability with which to round variables for each of these intervals. We denote by $\alpha(\obl)$ the approximation ratio achieved by this algorithm. (See \cref{sec:formal-defns} below for formal definitions of these objects.)

To state the theorems properly, we first define some relevant quantities which first arose in the context of small-space sketching algorithms for $\mkand$ in the work of Boyland \emph{et al.}~\cite{BHP+22}. We define

\begin{equation}\label{eq:gamma-k}
    \gamma_k \eqdef \begin{cases} \frac1k & k\text{ odd} \\ \frac1{k+1} & k\text{ even}, \end{cases}
\end{equation}
and
\begin{equation}
    p^*_k \eqdef \frac12(1+\gamma_k)
\end{equation}
and
\begin{equation}
    \alpha^*_k \eqdef 2 \cdot (p^*_k (1-p^*_k))^{\lfloor k/2 \rfloor}.
\end{equation}
That is, for even $k$, $\alpha^*_k = 2^{-(k-1)} (1-1/(k+1))^{k/2} (1+1/(k+1))^{k/2}$, and for odd $k$, $\alpha^*_k = 2^{-(k-1)} (1-1/k)^{(k-1)/2} (1+1/k)^{(k-1)/2}$. In particular, at $k=2$, we have $p^*_k=2/3$ and $\alpha^*_k = 4/9$.

Our first theorem states that the optimal superoblivious algorithm for $\mkand$ achieves ratio $\alpha^*_k$, and that this algorithm rounds with probability $p^*_k$:

\begin{theorem}[Characterization for superoblivious algorithms]\label{thm:superobl}
    For every $k \geq 2$, there is a unique superoblivious algorithm $\obl$ achieving ratio $\alpha(\obl) = \alpha^*_k$, and all other superoblivious perform strictly worse. (In particular, for $\vect = (0,1)$, every rounding vector $\vecp = (p)$ satisfies $\alpha(\obl) \leq \alpha^*_k$, with equality if and only if $p = p^*_k$.)
\end{theorem}

Our main theorem then states that one can improve over these superoblivious algorithms using other oblivious algorithms, and indeed, it suffices to consider only slight ``perturbations'' of the optimal superoblivious algorithms:

\begin{theorem}[Main theorem: Better oblivious algorithms]\label{thm:perturb}
    For every $k \geq 2$, there exists a bias partition $\vect$ and a rounding vector $\vecp$ such that the oblivious algorithm $\obl$ achieves $\alpha(\obl) \gneq \alpha^*_k$. (In particular, there exists $\epsilon^* > 0$ such that for all $0 < \epsilon \leq \epsilon^*$, there exists $0 < \delta < 1$ such that $\vect = (\delta, 1)$ and $\vecp = (p^*_k+\epsilon)$ satisfy $\alpha(\obl) \gneq \alpha^*_k$.)
\end{theorem}

These theorems are both proven by analyzing the dual of a certain natural linear program. Arguably, this ``dual'' perspective systematizes the \emph{ad hoc} analyses of small-space sketching algorithms for $\mdcut$ in \cite{CGV20} and for $\mkand$ in \cite{BHP+22}. Indeed, the analysis in those paper examined certain systems of linear inequalities using ``elementary'' reasoning (i.e., taking nonnegative linear combinations), and it is exactly this type of reasoning which is captured by the technology of dual linear programs.

Our results also have the following implications for streaming algorithms, generalizing connections established by Saxena, Singer, Sudan, and Velusamy~\cite{SSSV23-random-ordering} for the special case of $\mdcut$. We say an instance $\Psi$ of $\mkand$ is in \emph{input form} if it is unweighted (i.e., the weight of every clause is $1$), though multiple copies of the same clause are allowed. These instances will be the input to our algorithms though this is essentially without loss of generality as general instances can be ``rounded'' to such instances via standard arguments.

\begin{theorem}[Random-order streaming algorithm]\label{cor:random-order}
    For all $k \geq 2$, there exists $\alpha > \alpha^*_k$ such that for all $\epsilon > 0$, the following holds. There is an $O(\log n)$-space streaming algorithm which, for every instance $\Psi$ of $\mkand$ in input form with $n$ variables and $\poly(n)$ clauses, given as input $\Psi$'s clauses in a \emph{randomly-ordered} stream, outputs an $(\alpha - \epsilon)$-approximation to the $\mkand$ value of $\Psi$ with probability $99/100$.
\end{theorem}

\begin{theorem}[Bounded-degree streaming algorithm]\label{cor:bounded-degree}
    For all $k \geq 2$, there exists $\alpha > \alpha^*_k$ such that for all $\epsilon > 0$, the following holds: For all $D \geq 2$, there is an $O(D^{1/k} n^{1-1/k} \log n/\epsilon^{2/k})$-space streaming algorithm which, for every instance $\Psi$ of input form with $n$ variables and \emph{maximum degree} $\leq D$ (i.e., every variable is contained in $\leq D$ clauses), given as input $\Psi$'s clauses in an \emph{adversarially-ordered} stream, outputs an $(\alpha^*_k+\epsilon)$-approximation to the $\mkand$ value of $\Psi$ with probability $99/100$.
\end{theorem}

Both of these results are interesting because, as shown by Boyland \emph{et al.}~\cite{BHP+22} (analyzing families of algorithms and lower bounds due to Chou \emph{et al.}~\cite{CGSV21-boolean}), there are $O(\log n)$-space streaming algorithms which output (arbitrarily close to) $\alpha^*_k$-approximations for adversarially-ordered streams, and this is the best achievable ratio in $o(\sqrt n)$ space.\footnote{Technically, this lower bound is currently only known to hold for a subclass of streaming algorithms called \emph{sketching} algorithms, but the algorithm in \cref{cor:bounded-degree} appears to be such an algorithm.} Therefore, our results show that by relaxing either the adversarial-ordering assumption or the space bound, one can achieve better algorithms (the latter under a bounded-degree assumption). Analogous results to \cref{cor:random-order,cor:bounded-degree} were obtained for $\mdcut$ by Saxena, Singer, Sudan, and Velusamy~\cite{SSSV23-dicut} for the special case of $\mdcut$, there contrasting with algorithms and lower bounds due to Chou, Golovnev, and Velusamy~\cite{CGV20}.

We also include some explicit improved approximation ratios calculated using computer search and LP solvers in \cref{tab:ratios}.

\begin{table}
\centering
\begin{tabular}{c|c|c|c|c|c}
    $k$ & \textbf{Upper bound} & \textbf{Superobl.} & \textbf{Prev.} & \textbf{New, pert.} & \textbf{New, piecewise lin.} \\ \hline
    $2$ & $1/2 = 0.5$ & $4/9\approx 0.4444$ & $0.4835$ \cite{FJ15} &  $0.4457$ & $0.4844$ @ $(200,0.5,1.0)$ \\ \hline
    $3$ & $1/4 = 0.25$ & $2/9 \approx 0.2222$ & & $0.2226$ & $0.2417$ @ $(30,0.7,1.0)$ \\ \hline
    $4$ & $1/8 = 0.125$ & $72/625 = 0.1152$ & & $0.1157$ & $0.1188$ @ ($11, 0.8, 0.8$) \\ \hline
    $5$ & $1/16 = 0.0625$ & $36/625 \approx 0.0576$ & & $0.0578$ & $0.0589$ @ $(7,0.95,0.8)$ \\
\end{tabular}
\caption{The above table displays concrete approximation ratio for $\mkand$, for $k \in \{2,\ldots,5\}$. In the second column, we write the trivial upper bound of $2^{-k}$ on the approximation ratio of all oblivious algorithms (and more generally all ``local'' algorithms) for $\mkand$ (see \cref{obs:upper-bound} below). In the third column, we write $\alpha^*_k$, the approximation ratio achieved by the best superoblivious algorithm (and also by the best $o(\sqrt n)$-space sketching algorithm \cite{BHP+22}). In the fourth column, we highlight that the only previously known CSP for which oblivious algorithms outperformed superoblivious algorithms was $\mtwoand$ (from \cite{FJ15}). In the fifth, we include approximation ratios from our ``perturbed superoblivious'' algorithms, as in \cref{thm:perturb} with $\delta = 0.01$ and $\epsilon = 0.001$. Finally, in the sixth column, we report ratios achieved by much more complex algorithms which we constructed. These algorithms are parametrized by triples $(\ell,x,y)$, where $\ell$ specifies the number of bias classes, and $x,y$ specify a rounding vector. $\ell$ is chosen such that the number of variables, which is roughly $(2\ell)^k$ (see \cref{sec:patterns} below), is in at most the hundreds of thousands, in order for the LP solver to run in a reasonable amount of time. The bias partition is a uniform partition of $[0,1]$ into $\ell$ intervals and, imitating the algorithm in \cite[Proof of Theorem 1.3]{FJ15}, our rounding vector are ``two-piece piecewise-linear functions'': the first part of the vector, up to bias $x$, interpolates linearly between probability $\frac12$ and $y$, and the second part interpolates linearly between probability $y$ and $1$. The values in the last two columns were calculated with a Python script and the LP solver \texttt{glpk}; the code is available online at \url{https://github.com/singerng/oblivious-csps}. For the final column, the parameters $(x,y)$ were chosen using a grid search; solving the final LPs took 1 hour, 56 minutes on a 2021 Macbook Pro.}\label{tab:ratios}
\end{table}

\subsection{Technical overview}

The first main technical step in the paper is to develop, for each oblivious algorithm $\obl$ (defined by a bias partition $\vect$ and a rounding vector $\vecp$), a linear program (LP) which characterizes the approximation ratio of $\obl$; this LP is contained in \cref{lem:primal-opt} below, and is a generalization of the LP developed for oblivious algorithms in $\mdcut$ in \cite{FJ15}. The LP has a simple structure: Each feasible solution corresponds to a certain family of instances of $\mkand$ on which $\obl$ produces the same approximation to the $\mkand$ value, and the objective equals this approximation value. In particular, we will assign to each clause in an instance $\Psi$ of $\mkand$ a ``pattern'' based on the biases of and negations on the variables, and use the observation that the probability any clause is satisfied depends only on this pattern.

Next, we formulate the dual LP for this original ``primal'' LP (see \cref{lem:dual-opt} below). Here is where we benefit massively from the fact that the performance of oblivious algorithms is captured by a linear program, because by the magic of LP duality, it is possible to constructively show that oblivious algorithms perform \emph{well}: While feasible solutions to the \emph{primal} LP \emph{upper-bound} the ratio achieved by an oblivious algorithm $\obl$, feasible solutions to the \emph{dual} LP \emph{lower-bound} the ratio! In other words, to prove that an oblivious algorithm performs well on \emph{all} instances, it suffices to construct a \emph{single} feasible dual solution.

To prove \cref{thm:perturb}, we now want to compare the dual LP for superoblivious algorithms and their ``perturbations'', and show that we can get ``improved'' feasible solutions in the latter case. It turns out that in this setting, the primal LP has $O(k^5)$ variables and only $7$ inequality constraints; therefore the dual LP has $7$ variables and $O(k^5)$ inequality constraints. We make the crucial observation that in the superoblivious case there is an optimal dual solution which is \emph{sparse}: It is supported on only $3$ variables. Since this dual solution is so simple, we can analytically prove its feasibility for all $k$ by establishing a certain ``two-sided Bernoulli's inequality'' (\cref{lem:bern-fancy} below). And moreover, this inequality will show that in the superoblivious case, when we plug in our special solution, all but $6$ of the $O(k^5)$ dual constraints have slack! Thus, for very small values of $\delta$, it will be sufficient to slightly perturb this special solution in a way that makes these $6$ ``core'' constraints strictly satisfied, and this is precisely what we do in \cref{lem:core-strict-even,lem:core-strict-odd} below. This involves careful analysis based on certain elementary inequalities, using simple inequalities such as that $\frac{1-\epsilon}{1-k\epsilon} > \frac{1+\epsilon}{1+k\epsilon}$ for all $\epsilon > 0$ and $k > 1$.

\subsection{Future questions}

\paragraph{Streaming algorithms.} We hypothesize that the bounded-degree assumption in \cref{cor:bounded-degree} can be relaxed to give an $\tilde{O}(n^{1-1/k})$-space algorithm for \emph{all} instances (in input form with $\poly(n)$ clauses). Specifically, Saxena, Singer, Sudan, and Velusamy~\cite{SSSV23-dicut} developed sketching techniques enabling such a guarantee for $\mdcut$, the bounded-degree counterpart being provided by their earlier work \cite{SSSV23-random-ordering}; perhaps these ideas can be extended to $\mkand$.

\paragraph{More CSPs.} It would also be interesting to extend the framework in this paper to more CSPs, both other Boolean CSPs and to CSPs over larger alphabets. To the best of our knowledge, it is even plausible that every CSP which admits nontrivial $O(\log n)$-space sketching algorithms (as analyzed in \cite{CGSV21-finite}) also admits ``oblivious-style'' approximation algorithms, which in turn yield better sublinear-space streaming algorithms. A good starting point here would be to analyze symmetric Boolean CSPs, since in that setting we know that all CSPs which do not support one-wise distributions of satisfying assignments admit such nontrivial sketching algorithms (see \cite[Proposition 2.10]{CGSV21-boolean}); this class includes $\mkand$, for which we developed such results in this paper, but we could hope for improved ``oblivious-style'' algorithms for other such CSPs, e.g. symmetric threshold functions.

\paragraph{``Uniform'' hard instances.} In the special case of $\mdcut$, Feige and Jozeph~\cite{FJ15} constructed what might be called a \emph{uniformly hard} instance of $\mdcut$: For this single instance, \emph{every} oblivious algorithm achieves a ratio less than than $0.49$. (This strengthens the $\frac12$  bound from a ``trivial'' instance, a single bidirected edge; see also \cref{obs:upper-bound} below.) It would be interesting to construct similar instances for $\mkand$, $k \geq 3$, especially if such this construction could be made analytic. We note that such an object corresponds to a feasible solution to the linear program in \cref{lem:primal-opt} for which \emph{every} choice of rounding vector has objective strictly less than $\frac12$, and therefore for $\mkand$, any proof would require certifying that a certain degree-$k$ polynomial is bounded below $\frac12$ over $p \in [0,1]$.

\paragraph{An optimal rounding curve?} To construct an $0.483$-approximate oblivious algorithm for $\mdcut$, Feige and Jozeph~\cite{FJ15} rounded vertices using (a discretization of) a sigmoid-shaped piecewise-linear function: This function rounds vertices with bias $b \in (0,1]$ to $1$ with probability \[ p(b) = \begin{cases} \frac12+b & 0 \leq b \leq \frac12 \\ 1 & \frac12 \leq b \leq 1\end{cases}. \] But is it possible to analytically calculate the optimal rounding function (and is it unique)? Given any discretization of biases into intervals, one could in principle enumerate all basic feasible solutions to the LP, and then calculate the best rounding vector; that is, each rounding vector will induce an objective function for the LP, and the best rounding vector maximizes the minimum objective over all basic feasible solution. Towards this, it might be helpful to get a handle on the vertices of this LP's polytope, and whether there is some simple way to enumerate them.

\subsection*{Outline}

We define $\mkand$ and oblivious algorithms formally in \cref{sec:formal-defns} and develop the linear-programming characterization for the approximation ratio of oblivious algorithms, and some other basic tools, in \cref{sec:lp}. We analyze the dual LP to prove \cref{thm:perturb} in \cref{sec:thm:perturb}. We prove \cref{thm:superobl}, this time by analyzing the primal LP, in \cref{sec:thm:superobl}. Finally, we prove our theorems on streaming algorithms (\cref{cor:bounded-degree,cor:random-order}) in \cref{sec:streaming}.

\section{Definitions: $\mkand$ and oblivious algorithms}\label{sec:formal-defns}

We now give formal definitions for the $\mkand$ problem and for oblivious algorithms. For the remainder of the paper, we adopt a (nonstandard) convention which views variables in $\mkand$ as taking $\{-1,+1\}$ values; this is for notational convenience in defining bias and similar concepts.

\begin{definition}[$\mkand$]
An \emph{instance} of the $\mkand$ problem on $n$ variables is given by a sequence of constraints $C_1,\ldots,C_m$, with $C_j = (V^+_j,V^-_j,w_j)$, consisting of ``positive variables'' $V^+_j \subseteq [n]$ and ``negative variables'' $V^-_j \subseteq [n]$ with $V^+_j \cap V^-_j = \emptyset$ and $|V^+_j \cup V^-_j| = k$, and a weight $w_j \geq 0$. An \emph{assignment} for this problem is given by $(\vecx) = (x_1,\ldots,x_n) \in \{\pm1\}^n$, and the \emph{value} of this assignment is \[ \val_\Psi(\vecx) \eqdef \frac{\sum_{j=1}^m \1[x_v=+1 \; \forall v \in V_j^+ \wedge x_v = -1 \; \forall v \in V_j^-] \cdot w_j}{\sum_{j=1}^m w_j}. \] The \emph{value} of the instance $\Psi$ is \[ \val_\Psi \eqdef \max_{\vecx \in \{\pm1\}^n} \val_\Psi(\vecx). \]
\end{definition}

Next, towards defining the bias of a variable in an instance, for any variable $v \in [n]$, we define its \emph{positive} and \emph{negative weight}:

\begin{equation}
    w^+_\Psi(v) \eqdef \sum_{j=1}^m \1[v \in V^+_j] \cdot w_j \text{ and } w^-_\Psi(v) \eqdef \sum_{j=1}^m \1[v \in V^-_j] \cdot w_j.
\end{equation}

Then, we define the \emph{bias} of a variable as:\footnote{Throughout the paper, we assume every variable appears in at least one constraint, and therefore that $w^+_\Psi(v)+w^-_\Psi(v) > 0$. (This is WLOG, since variables appearing in no constraints can be ignored for the purposes of $\mkand$.)}
\begin{equation}
    \bias_\Psi(v) \eqdef \frac{w^+_\Psi(v)-w^-_\Psi(v)}{w^+_\Psi(v)+w^-_\Psi(v)}.
\end{equation}

Next, we consider \emph{symmetric} ways to partition the space of possible biases $[-1,1]$ into $L=2\ell+1$ intervals labeled by $\{-\ell,\ldots,+\ell\}$. The data of such a partition is a ``bias partition'' vector $\vect = (t_0,\ldots,t_\ell)$ with $0 \leq t_0 < \cdots < t_{+\ell}=1$. We denote the $i$-th interval by $\I_i$, and let for $i \geq 1$, we let $\I_{+i} = (t_{i-1},t_i]$ and $\I_{-i} = [-t_i,-t_{i-1})$; and the $0$-th interval be the center $[-t_0,t_0]$. Thus, the intervals $\I_{-\ell},\ldots,\I_0,\ldots,\I_{+\ell}$ partition the interval $[-1,1]$ of possible biases.\footnote{Our choice of which ends of these intervals are open and which are closed is an arbitrary convention; the only important property of the decomposition of $[-1,1]$ into intervals is that it is \emph{symmetric}.} For notational convenience, we let $t^+_i$ and $t^-_i$ denote the upper and lower bounds $\sup \I_i$ and $\inf \I_i$, respectively. (So, e.g., $t_i^+ = t_i$ for $i \geq 0$, whereas $t_i^+ = t_{-(i+1)}$ for $i < 0$.) We also consider symmetric ways to round vertices based on these classes. The data of such a rounding scheme is a ``rounding vector'' $\vecp=(p_1,\ldots,p_\ell)$ of probabilities. Given these, we can define ``$L$-class'' oblivious algorithms:

\begin{definition}[Oblivious algorithm for $\mkand$]
Let $L = 2\ell+1 \geq 3$ be an odd integer. Let $\vect = (t_0,\ldots,t_\ell)$ be an bias partition and $\vecp = (p_1,\ldots,p_\ell)$ a rounding vector. For any $k \geq 2$, the oblivious algorithm $\obl$ for $\mkand$ behaves as follows: Given an instance $\Psi$, for each variable $v \in \{1,\ldots,n\}$ independently:

\begin{itemize}
    \item If $\bias_\Psi(v) \in \I_0$, assign $x_v \mapsto 1$ w.p. $\frac12$, $x_v \mapsto -1$ w.p. $\frac12$.
    \item If $\bias_\Psi(v) \in \I_{+i}$ for $i \geq 1$, assign $x_v \mapsto 1$ w.p. $p_i$, $x_v \mapsto -1$ w.p. $1-p_i$.
    \item If $\bias_\Psi(v) \in \I_{-i}$ for $i \geq 1$, assign $x_v \mapsto 1$ w.p. $1-p_i$, $x_v \mapsto -1$ w.p. $p_i$.
\end{itemize}

We denote by $\obl(\Psi)$ the expected value of the assignment produced by this rounding scheme,\footnote{We abuse this notation and often think of $\obl(\Psi)$ as the \emph{output} of the oblivious algorithm, i.e., we think of the oblivious algorithm's goal as outputting a (scalar) estimate of the value of the instance; this holds especially in the context of streaming algorithms.} and by \[ \alpha(\obl) \eqdef \inf_\Psi \frac{\obl(\Psi)}{\val_\Psi} \] the approximation ratio achieved by this algorithm.
\end{definition}

In the simplest interesting case, we have $\ell=1$, $\vect=(0,1)$, and $\vecp=(p)$. These algorithms, which we call \emph{superoblivious} algorithms, ignore the \emph{magnitude} of the bias of each variable, rounding only based on sign: E.g., negatively-biased variables are rounded to $1$ w.p. $1-p$.

We remark that there are a few natural ways to generalize this definition of oblivious algorithms. Firstly, we could consider rounding functions which are not ``antisymmetric'', i.e., we could round bias-$(+b)$ and bias-$(-b)$ variables with probabilities which are not complementary. In particular, for $b=0$, we could round bias-$0$ variables could be rounded with non-uniform probability; however, such an algorithm would strictly underachieve any antisymmetric algorithm on simple instances (see \cref{obs:upper-bound} below). Also, we could use continuous rounding functions instead of breaking up the range of biases into discrete intervals, but such an algorithm would not be amenable to analysis of the approximation ratio by a linear program.

\begin{observation}\label{obs:upper-bound}
There is a simple lower-bound construction which shows that \emph{no} oblivious algorithm for $\mkand$ can achieve a ratio better than $2^{-(k-1)}$. (Note that the optimal superoblivious ratio $\alpha^*_k$ equals this upper bound times a ``discounting'' factor.) Consider any $k$: In the instance with two equally weighted constraints, $C^+ = (+1,\ldots,+1),(1,\ldots,k)$ and $C^- = (-1,\ldots,-1),(1,\ldots,k)$, i.e., the two constraints want $x_1,\ldots,x_k$ to be all-$(+1)$'s and all-$(-1)$'s, respectively. Every variable has bias zero so it will be rounded uniformly by every oblivious algorithm, yielding value $2^{-k}$, while the ``greedy'' all-$(+1)$'s (or all-$(-1)$'s) assignment achieves value $\frac12$. Indeed, this ``lower bound'' holds for any class of algorithms which cannot ``break the symmetry'' between these two greedy assignments.
\end{observation}

\section{The linear-programming framework for oblivious algorithms}\label{sec:lp}

In this section, we develop a linear program which captures the ``worst-case instance'' for any oblivious algorithm, and therefore can be used to calculate the approximation ratio (\cref{lem:primal-opt}), as well as the corresponding dual linear program (\cref{lem:dual-opt}). These will be applied to bound the approximation ratios of certain oblivious algorithms in the following sections.

\subsection{Clause patterns}\label{sec:patterns}

Let $\Clauses^L_k$ denote the set of vectors $\vecc=(\vecc^+,\vecc^-)=(c^+_{-\ell},\ldots,c^+_{+\ell},c^-_{-\ell},\ldots,c^-_{-\ell})$ whose entries are natural numbers and sum to $k$. These are useful because they describe each particular clause from the perspective of an $L$-class oblivious algorithms. In particular, given a clause $C$, we denote its \emph{pattern} $\ptn(C) = (c^+_{-\ell},\ldots,c^+_{+\ell},c^-_{-\ell},\ldots,c^-_{-\ell}) \in \Clauses^L_k$ where $c^+_i$ and $c^-_i$ denote the number of positive and negative literals in $C$ whose variables have bias class $i$, respectively, for each $i \in \{-\ell,\ldots,+\ell\}$. That is, e.g., \[ c^+_i = |\{v \in V^+_j : \bias_\Psi(v) \in \I_i \}|. \]

Now, for any rounding vector $\vecp = (p_1,\ldots,p_\ell)$, we define

\begin{equation}\label{eq:p}
    \prob^\vecp(\vecc) = 2^{-(c^+_0+c^-_0)} \prod_{i=1}^\ell p_i^{c^+_{+i} + c^-_{-i}} (1-p_i)^{c^-_{+i} + c^+_{-i}}
\end{equation}
for each $\vecc \in \Clauses^L_k$.\footnote{In this expression we adopt the convention $0^0=1$, i.e., if $p_i=0$ but $c^-_{+i} + c^+_{-i}=0$ then we ignore the factor $0$.} Then we have:

\begin{claim}\label{claim:obl-exp}
Let $\Psi$ be an instance of $\mkand$ with clauses $C_1,\ldots,C_m$ with weights $w_1,\ldots,w_m$, respectively. Then \[ \obl(\Psi) = \frac{\sum_{j=1}^m \prob^\vecp(\ptn(C_j)) \cdot w_j}{\sum_{j=1}^m w_j}. \]
\end{claim}

\begin{proof}
    By linearity of expectation, it suffices to show that each clause $C_j$ is satisfied w.p. $\prob^\vecp(\ptn(C_j))$. We can rewrite \[ \prob^\vecp(\vecc) = 2^{-c^+_0}  2^{-c^-_0} \prod_{i=1}^\ell p_i^{c^+_{+i}} p_i^{c^-_{-i}} (1-p_i)^{c^-_{+i}} (1-p_i)^{c^+_{-i}}. \] Recalling that each variable is assigned independently, and the clause is satisfied iff each literal is, the above expression precisely represents the probability that the clause is satisfied. (E.g., if there is a negative literal whose variable has bias class $+i$, this literal is satisfied with probability $1-p_i$; the number of such factors in the probability is $c^-_{+i}$.)
\end{proof}

We observe that $|\Clauses_k^L| = \binom{k+2L-1}{2L-1}$ by the ``stars-and-bars'' formula. For instance, if $L=3$ (as will be the case in the explicit analysis in the following sections), we have $|\Clauses_k^L| = O(k^5)$.

\subsection{The factor-revealing linear program}

We denote by $\PClauses^L_k \subseteq \Clauses^L_k$ the space of clause patterns without negations, i.e., $\vecc$ such that $c^-_{-\ell}=\cdots=c^-_{+\ell}=0$. For two vectors $\vecx = (x_1,\ldots,x_n), \vecy = (y_1,\ldots,y_n) \in \BR^n$, let $\vecx \odot \vecy = (x_1y_1,\ldots,x_ny_n)$ denote their entrywise product. To design the linear program, we will need the following useful proposition:

\begin{proposition}[Flipping]\label{claim:flipping}
Let $\Psi$ be an instance of $\mkand$, and for any assignment $\vecy=(y_1,\ldots,y_n) \in \{\pm1\}^n$, let $\flip^\vecy(\Psi)$ denote the instance of $\mkand$ where we ``flip'' the variables $v$ with $y_v = -1$; that is, each clause $C_j = (V^+_j,V^-_j,w_j)$ in $\Psi$ becomes a clause $D_j = (U^+_j,U^-_j,w_j)$ where $U^+_j = \{v \in V^+_j : y_v = +1\} \cup \{v \in V^-_j : y_v = -1\}$ and $U^-_j = \{v \in V^+_j : y_v = -1\} \cup \{v \in V^-_j : y_v = +1\}$. Then:

\begin{itemize}
    \item For every assignment $\vecx \in \{\pm1\}^n$, $\val_\Psi(\vecx) = \val_{\flip^\vecy(\Psi)}(\vecx \odot \vecy)$.
    \item In particular, if $\vecx$ is an optimal assignment to $\Psi$, then $\vecx\odot\vecy$ is an optimal assignment to $\flip^\vecy(\Psi)$.
    \item $\obl(\Psi) = \obl(\flip^\vecy(\Psi))$.
\end{itemize}
\end{proposition}

\begin{proof}
    Follows immediately from definitions.
\end{proof}

\begin{lemma}[Primal characterization]\label{lem:primal-opt}
For every bias partition $\vect = (t_0,\ldots,t_\ell)$ and rounding vector $\vecp = (p_1,\ldots,p_\ell)$, the approximation ratio $\alpha(\obl)$ achieved by $\obl$ equals the value of the following linear program:

\begin{empheq}[left=\empheqlbrace]{alignat*=3}
    & \mathrm{minimize} \quad && \sum_{\vecc \in \Clauses^L_k} \prob^\vecp(\vecc) \cdot W(\vecc)  && \\
    & \mathrm{s.t.} && W(\vecc) \geq 0 && \forall \vecc \in \Clauses^L_k \\
    & && \sum_{\vecc \in \PClauses^L_k} W(\vecc) = 1 && \\
    & && t^-_i (W^+(i) + W^-(i)) \leq W^+(i) - W^-(i) \quad && \forall i \in \{-\ell,\ldots,+\ell\} \\
    & && W^+(i) - W^-(i) \leq t^+_i (W^+(i) + W^-(i)) \quad && \forall i \in \{-\ell,\ldots,+\ell\}
\end{empheq}
where we define the linear functions \[ W^+(i) = \sum_{\vecc \in \Clauses^L_k} c^+_i W(\vecc) \text{ and } W^-(i) = \sum_{\vecc \in \Clauses^L_k} c^-_i W(\vecc). \]
\end{lemma}

\begin{proof}
    Let $\aalg$ denote the approximation ratio of $\obl$, and $\alp$ the minimum value of the linear program. This proof generalizes \cite[Proof of Theorem 1.2]{FJ15}.

    ($\alp \leq \aalg$) We show that for every instance $\Psi$ of $\mkand$, there is a feasible LP solution $\{W(\vecc)\}_{\vecc \in \Clauses_k^L}$ of objective value $\obl(\Psi)/\val_\Psi$.
    
    Towards this claim, by \cref{claim:flipping}, we can assume WLOG that the all-$(+1)$'s assignment is optimal for $\Psi$. Also, we assume WLOG by rescaling that $\Psi$ has total weight $\frac1{\val_\Psi}$, i.e., $\sum_{j=1}^m w_j = \frac1{\val_\Psi}$. Now, let $C_1,\ldots,C_m$ denote the constraints of $\Psi$, and let $W(\vecc) := \sum_{j=1}^m \1[\ptn(C_j)=\vecc] w_j$. We claim that $\{W(\vecc)\}_{\vecc \in \Clauses_k^L}$ is feasible and has objective value $\sum_{\vecc \in \Clauses_k^L} \prob^\vecp(\vecc) W(\vecc) = \obl(\Psi)/\val_\Psi$.
    
    First, we check feasibility. Clearly all $W(\vecc)$'s are nonnegative. Next, we have
    \begin{align*}
        \val_\Psi &= \val_\Psi(\vecpos) \tag{all-$(+1)$'s is optimal} \\
        &= \frac{\sum_{j=1}^m w_j \1[|V^-_j|=0]}{\sum_{j=1}^m w_j} \tag{def. of $\val_\Psi(\vecpos)$} \\
        &= \frac{\sum_{\vecc \in \PClauses^L_k} W(\vecc)}{1/\val_\Psi} \tag{def. of $\PClauses^L_k$ and total weight assumption}
    \end{align*}
    which rearranges to $\sum_{\vecc \in \PClauses^L_k} W(\vecc)=1$.
    
    Now, recall the definitions of $\bias_\Psi, w^+_\Psi, w^-_\Psi$ from \cref{sec:formal-defns}. Fix a bias class $i \in \{-\ell,\ldots,+\ell\}$. For any variable $v$ with bias class $i$, we have $\bias_\Psi(v) \in \I_i$, so $t^-_i \leq \bias_\Psi(v) \leq t^+_i$, so multiplying through by $w^+_\Psi(v) + w^-_\Psi(v)$, we get \[ t^-_i (w^+_\Psi(v) + w^-_\Psi(v)) \leq w^+_\Psi(v) - w^-_\Psi(v) \leq t_i^+ (w^+_\Psi(v) + w^-_\Psi(v)). \] Letting $\CV_i$ denote the set of all variables in $\Psi$ with bias class $i$, we can sum over these equations to get \[ t^-_i \sum_{v\in\CV_i} (w^+_\Psi(v) + w^-_\Psi(v)) \leq \sum_{v\in\CV_i} (w^+_\Psi(v) - w^-_\Psi(v)) \leq t_i^+ \sum_{v\in\CV_i}(w^+_\Psi(v) + w^-_\Psi(v)). \] We claim that \[ W^+(i) = \sum_{v\in\CV}w^+_\Psi(v), \] and similarly $W^-(i) = \sum_{v \in \CV}w^-_{\Psi}(v)$. These equalities imply that $W(\cdot)$ satisfies the feasibility constraints, and it remains to prove them. Now recall $w^+_\Psi(v) = \sum_{j=1}^m \1[v \in V^+_j] w_j$; therefore, \[ \sum_{v\in\CV} w^+_\Psi(v) = \sum_{j=1}^m \sum_{v \in V^+_j} \1[\bias_\Psi(v)\in\I_i]  w_j, \] and $j$-th term in this sum is precisely $c^+_i$ where $\vecc = (\vecc^+,\vecc^-) = \ptn(C_j)$. The proof for $W^-(i)$ is similar.
    
    Finally, by \cref{claim:obl-exp} and our assumption $\sum_{j=1}^m w_j = 1/\val_\Psi$, we have \[ \obl(\Psi) = \frac{\sum_{j=1}^m w_j \prob^\vecp(\ptn(C_j))}{\sum_{j=1}^m w_j} = \frac{\sum_{\vecc \in \Clauses^L_k} W(\vecc) \cdot \prob^\vecp(\vecc)}{1/\val_\Psi}, \] which rearranges to $\obl(\Psi)/\val_\Psi = \sum_{\vecc \in \Clauses^L_k} W(\vecc) \cdot \prob^\vecp(\vecc)$, as desired.

    ($\alp \geq \aalg$) This argument is essentially converse to the former argument, but there are two technical issues: (i) the linear program does not encode strict inequality constraints, while an oblivious algorithm needs to (in the sense that e.g., if $t_0 = 0$, then the algorithm rounds vertices with bias $0$ and bias $+\epsilon$ differently), and (ii) since an $\mkand$ constraint cannot use a variable twice, we might need many variables with the same bias in the instance we create.

    In our argument, we define a property of certain feasible solutions called ``niceness'', and show that (1) for every feasible LP solution $\{W(\vecc)\}_{\vecc \in \Clauses_k^L}$ of objective $v$, for all $\epsilon > 0$, there is a \emph{nice} feasible solution $\{W'(\vecc)\}_{\vecc \in \Clauses_k^L}$ of objective $\leq v + \epsilon$, and (2) for every nice feasible solution $\{W(\vecc)\}_{\vecc \in \Clauses_k^L}$ of objective $v$, there is an instance $\Psi$ of $\mkand$ where $\obl(\Psi) / \val_\Psi \leq v+\epsilon$. Together, these imply that for every feasible LP solution $\{W(\vecc)\}_{\vecc \in \Clauses_k^L}$ with objective $v$, then for all $\epsilon > 0$, there is an instance $\Psi$ of $\mkand$ where $\obl(\Psi) / \val_\Psi \leq v + \epsilon$, and this suffices.
    
    Towards (1), let $\{W(\vecc)\}_{\vecc \in \Clauses_k^L}$ be any feasible solution with objective value $v$. Our notion of ``niceness'' is: For all $i \neq 0 \in \{-\ell,\ldots,+\ell\}$, the hypothesized ``bias'' inequalities are strict, i.e., \[ t^-_i (W^+(i) + W^-(i)) <  W^+(i) - W^-(i) < t^+_i (W^+(i) + W^-(i)). \] (We exclude $i=0$ because we could have $t^-_0=t^+_0=0$, but for other $i$ our definition of bias partitions implies $t^-_i < t^+_i$.) To construct nice $\{W'(\vecc)\}_{\vecc \in \Clauses_k^L}$ from $\{W(\vecc)\}_{\vecc \in \Clauses_k^L}$, for each $i \in \{-\ell,\ldots,+\ell\} \setminus\{0\}$, if $t^-_i (W^+(i) + W^-(i)) =  W^+(i) - W^-(i)$ then we set $W'(\vecc) \gets W(\vecc) + \epsilon$ where $\vecc$ has $c^+_i=k$ and zeros elsewhere, and similarly if $W^+(i) - W^-(i) = t^+_i (W^+(i) + W^-(i))$ we set $W'(\vecc) \gets W(\vecc)+\epsilon$ where $\vecc$ has $c^-_i=k$ and zeros elsewhere; and we set $W'(\vecc) \gets W(\vecc)$ for all $\vecc$'s not already defined. Finally, we renormalize for the equality constraint, e.g., we set \[ W'(\vecc) \gets \frac{W'(\vecc)}{\sum_{\vecc' \in \PClauses_k^L} W'(\vecc)}. \] Observe that $\{W'(\vecc)\}$ is by definition nice (since renormalizing preserves strictness in the inequalities), feasible (for sufficiently small $\epsilon$), and further (reparametrizing $\epsilon$) we can preserve the objective up to arbitrarily small error.

    Now for (2), for any nice feasible solution $\{W(\vecc)\}_{\vecc \in \Clauses_k^L}$, we construct an instance $\Psi$ with $n = Lk$ variables such that $\obl(\Psi)/\val_\Psi \leq v$. These variables are labeled with tuples in $\CI:=\{-\ell,\ldots,+\ell\} \times [k]$, i.e., there are $k$ variables corresponding to each bias class. For every pattern $\vecc = (c^+_{-\ell},\ldots,c^+_{+\ell},c^-_{-\ell},\ldots,c^-_{+\ell}) \in \Clauses_k^L$, we let $\CJ_\vecc$ denote the set of $2L$-tuples of sets $(J^+_{-\ell},\ldots,J^+_{+\ell},J^-_{-\ell},\ldots,J^-_{+\ell})$ such that for each $i \in \{-\ell,\ldots,+\ell\}$ we have $J^+_i \subseteq \{i\}\times[k]$ and $J^-_i\subseteq \{i\}\times[k]$, $|J^+_i|=c^+_i$ and $|J^-_i|=c^-_i$, and $J^+_i\cap J^-_i = \emptyset$. Now for each such $2L$-tuple $J \in \CJ_\vecc$ we create a clause in $\Psi$, labeled $C_{\vecc,J}$, with $V_{\vecc,J}^+ = J^+_{-\ell} \cup \cdots J^+_{+\ell}$ and $V_{\vecc,J}^- = J^-_{-\ell} \cup \cdots J^-_{+\ell}$, and weight \[ w_{\vecc,J} = \frac{W(\vecc)}{|\CJ_\vecc|}. \]

    Now we observe that in this instance, the value of the all-$(+1)$'s assignment is
    \begin{align*}
        \val_\Psi(\vecpos) &= \frac{\sum_{\vecc \in \Clauses_k^L} \sum_{J \in \CJ_\vecc} \1[V^-_{\vecc,J}=\emptyset] w_{\vecc,J}}{\sum_{\vecc \in \Clauses_k^L} \sum_{J \in \CJ_\vecc} w_{\vecc,J}} \tag{def. of $\val$} \\
        &= \frac{\sum_{\vecc \in \Clauses_k^L} \sum_{J \in \CJ_\vecc} \1[V^-_{\vecc,J}=\emptyset] W(\vecc) / |\CJ_\vecc|}{\sum_{\vecc \in \Clauses_k^L} \sum_{J \in \CJ_\vecc} W(\vecc) / |\CJ_\vecc|} \tag{def. of $w_{\vecc,J}$} \\
        &= \frac{\sum_{\vecc \in \Clauses_k^L} \1[\vecc \in \PClauses_k^L] \sum_{J \in \CJ_\vecc} W(\vecc) / |\CJ_\vecc|}{\sum_{\vecc \in \Clauses_k^L} \sum_{J \in \CJ_\vecc} W(\vecc) / |\CJ_\vecc|} \tag{def. of $V^-_{\vecc,J}$ and $J$} \\
        &= \frac{\sum_{\vecc \in \Clauses_k^L} \1[\vecc \in \PClauses_k^L] W(\vecc)}{\sum_{\vecc \in \Clauses_k^L} W(\vecc)} \tag{summing constant} \\
        &= \frac{1}{\sum_{\vecc \in \Clauses_k^L} W(\vecc)} \tag{$W(\cdot)$ is feasible}.
    \end{align*}
    Therefore $\val_\Psi \geq \frac{1}{\sum_{\vecc \in \Clauses_k^L} W(\vecc)}$.

    Next, we claim that for every variable $(i,a) \in \CI$, $w^+_\Psi(i,a) = W^+(i)/k$ and $w^-_\Psi(i,a)=W^-(i)/k$. Assuming this, we will have that for all $(i,a) \in \CI$, $\bias_\Psi(i,a) = (W^+(i) - W^-(i))/(W^+(i) + W^-(i))$, and therefore by niceness $\bias_\Psi(i,a) \in \I_i$, and therefore that for every clause $C_{\vecc,j}$ in $\Psi$, $\ptn(C_{\vecc,j}) = \vecc$. Given this, by \cref{claim:obl-exp} and an analogous argument to the case of $\val_\Psi$, we have \[ \obl(\Psi) = \frac{\sum_{\vecc \in \Clauses_k^L} \sum_{J \in \CJ_\vecc} \prob^\vecp(\ptn(C_{\vecc,J})) w_{\vecc,J}}{\sum_{\vecc \in \Clauses_k^L} \sum_{J \in \CJ_\vecc} w_{\vecc,J}} = \frac{\sum_{\vecc \in \Clauses_k^L} \prob^\vecp(\vecc) W(\vecc)}{\sum_{\vecc \in \Clauses_k^L} W(\vecc)} = \frac{v}{\sum_{\vecc \in \Clauses_k^L} W(\vecc)}. \] Therefore $\obl(\Psi)/\val_\Psi \leq v$, as desired.
    
    Finally, it remains to prove the claim. This uses a counting argument. We prove $w^+_\Psi(i,a) = W^+(i)/k$; the proof for $W^-$ is analogous. We have:
    \begin{align*}
        w^+_\Psi(i,a) &= \sum_{\vecc \in \Clauses_k^L} \sum_{J \in \CJ_\vecc} \1[(i,a) \in V^+_{\vecc,J}] w_{\vecc,J} \tag{def. of $w^+$} \\
        &= \sum_{\vecc \in \Clauses_k^L} \sum_{J \in \CJ_\vecc} \1[(i,a) \in V^+_{\vecc,J}] W(\vecc)/|\CJ_\vecc| \tag{def. of $w$} \\
        &= \sum_{\vecc \in \Clauses_k^L} \sum_{J \in \CJ_\vecc} \1[a \in J^+_i] W(\vecc)/|\CJ_\vecc| \tag{def. of $C_{\vecc,J}$}.
    \end{align*}
    Therefore, since $W^+(i) = \sum_{\vecc \in \Clauses_k^L} c^+_i W(\vecc)$, it suffices to show that for all $a\in[k]$ and $\vecc \in \Clauses_k^L$, we have \[ \frac{|\{J \in \CJ_\vecc : a \in J^+_i\}|}{|\CJ_\vecc|} = \frac{c^+_i}{k}. \] Indeed, we have $|\CJ_\vecc| = \prod_{i'=-\ell}^{+\ell} \binom{k}{c^+_{i'}} \binom{k-c^+_{i'}}{c^-_{i'}}$ (since each $J \in \CJ_\vecc$, for each $i'$, independently chooses a disjoint pair of subsets from $[k]$, one of size $c^+_i$ and one of size $c^-_i$), and similarly \[ |\{J \in \CJ_\vecc : a \in J^+_i\}| = \binom{k-1}{c^+_i-1} \binom{k-c^+_i}{c^-_i} \prod_{i'\neq i \in \{-\ell,\ldots,+\ell\}} \binom{k}{c^+_{i'}} \binom{k-c^+_{i'}}{c^-_{i'}} \] (since our choices are the same for $i' \neq i$, but for $i$ we are forced to include $a$), and we can finally use the binomial identity $\binom{k}{n} = \frac{k}n \binom{k-1}{n-1}$.
\end{proof}

\begin{lemma}[Dual characterization]\label{lem:dual-opt}
For every bias partition $\vect = (t_0,\ldots,t_\ell)$ and rounding vector $\vecp = (p_1,\ldots,p_\ell)$, the approximation ratio $\alpha(\obl)$ achieved by $\obl$ equals the value of the following linear program:

\begin{empheq}[left=\empheqlbrace]{alignat*=4}
    & \mathrm{maximize} \quad && z && \\
    & \mathrm{s.t.} && \1[\vecc \in \PClauses^L_k] \cdot z && \\
    & && + \sum_{i=-\ell}^{+\ell} \left(((1-t^-_i)c_i^+-(t_i^-+1)c_i^-)y_i^- +((t_i^+-1)c_i^++(1+t_i^+)c_i^-)y_i^+\right) \leq \prob^\vecp(\vecc) \quad && &&\forall \vecc \in \Clauses^L_k \\
    & && y^-_i \geq 0 && && \forall i \in \{-\ell,\ldots,+\ell\} \\
    & && y^+_i \geq 0 && && \forall i \in \{-\ell,\ldots,+\ell\}
\end{empheq}
\end{lemma}

\begin{proof}
    To place the primal LP (from \cref{lem:primal-opt}) in a more standard form, we rewrite the primal inequality $t^-_i (W^+(i) + W^-(i)) \leq W^+(i) - W^-(i)$ as $(t^-_i-1)W^+(i) + (t^-_i+1)W^-(i) \leq 0$; expanding the definitions of $W^+(i)$ and $W^-(i)$, this is equivalent to $(t^-_i-1)\sum_{\vecc \in \Clauses^L_k} c^+_i W(\vecc) + (t^-_i+1)\sum_{\vecc \in \Clauses^L_k} c^-_i W(\vecc) \leq 0$. Similarly, the inequality $W^+(i) - W^-(i) \leq t^+_i (W^+(i) + W^-(i))$ becomes $(1-t_i^+)\sum_{\vecc \in \Clauses^L_k} c^+_i W(\vecc)-(1+t_i^+)\sum_{\vecc \in \Clauses^L_k} c^-_i W(\vecc) \leq 0$. Therefore, the primal LP is equivalent to the following standard-form LP:
    \begin{empheq}[left=\empheqlbrace]{alignat*=3}
        & \mathrm{minimize} \quad && \sum_{\vecc \in \Clauses^L_k} \prob^\vecp(\vecc) \cdot W(\vecc)  && \\
        & \mathrm{s.t.} && W(\vecc) \geq 0 && \forall \vecc \in \Clauses^L_k \\
        & && \sum_{\vecc \in \PClauses^L_k} W(\vecc) = 1 && \\
        & && (1-t_i^+)\sum_{\vecc \in \Clauses^L_k} c^+_i W(\vecc)-(1+t_i^+)\sum_{\vecc \in \Clauses^L_k} c^-_i W(\vecc) \leq 0 \quad && \forall i \in \{-\ell,\ldots,+\ell\} \\
        & && (t^-_i-1)\sum_{\vecc \in \Clauses^L_k} c^+_i W(\vecc) + (t^-_i+1)\sum_{\vecc \in \Clauses^L_k} c^-_i W(\vecc) \leq 0 \quad && \forall i \in \{-\ell,\ldots,+\ell\}
    \end{empheq}

    By LP duality, the above LP has the same value as its dual LP, which is the LP in the hypothesis.\footnote{See e.g. \cite[p. 85]{MG07}). One has to be careful with the signs, since our primal LP is a \emph{minimization} LP. Instead, we can consider the LP which maximizes $-\sum_{\vecc \in \Clauses_k^L} \prob^\vecp(\vecc) \cdot W(\vecc)$ (whose output is the negation of our desired output). Applying duality to this LP gives one which minimizes $z'$ such that $\1[\vecc \in \PClauses^L_k] \cdot z' + \sum_{i=-\ell}^\ell \left(((1-t^-_i)c_i^+-(t_i^-+1)c_i^-)y_i^- +((t_i^+-1)c_i^++(1+t_i^+)c_i^-)y_i^+\right) \geq -\prob^\vecp(\vecc)$. Transforming to a maximization problem equivalent to our original LP (since we had a negation!), we maximize $-z'$ such that $\1[\vecc \in \PClauses^L_k] \cdot z' + \sum_{i=-\ell}^\ell \left(((1-t^-_i)c_i^+-(t_i^-+1)c_i^-)y_i^- +((t_i^+-1)c_i^++(1+t_i^+)c_i^-)y_i^+\right) \geq -\prob^\vecp(\vecc)$. Finally, we negate both sides of this inequality, and use the bijective transformation $z=-z'$.}
\end{proof}

\section{Proving \cref{thm:perturb} by analyzing ``dual slack''}\label{sec:thm:perturb}

In this section, we prove \cref{thm:perturb} by constructing dual solutions which witness lower bounds on the approximation ratio of oblivious algorithms.

\subsection{A sufficient condition}

Our first step is the following lemma, which gives a clean sufficient condition for a lower bound on the approximation ratio by constructing a certain sparse dual solution and applying the dual program (\cref{lem:dual-opt})

\begin{lemma}[Sufficient conditions for good approximations]\label{lem:suff-cond}
For every $k \geq 2 \in \BN$, $0\leq\gamma,\delta \leq 1$, let $\vect=(\delta,1)$ and $\vecp=(\frac12(1+\gamma))$. The algorithm $\obl$ has approximation ratio $\alpha(\obl) \geq 2^{-(k-1)}\beta$ if the following statement holds: There exist $X,Y \geq 0$ such that:
    \begin{empheq}[left=\empheqlbrace]{align*}
    (1+\delta)\left(1-\frac{i+j}k\right)Y + (1-\delta) \frac{j}k X &\leq \beta^{-1} \left(1-\gamma\right)^i\left(1+\gamma\right)^j \quad \forall i,j\in\BN, i+j\leq k
    \\
    2 - (1-\delta) \left(1 - \frac{i+j}k\right) Y - (1+\delta)\frac{i}k X&\leq \beta^{-1} \left(1-\gamma\right)^i\left(1+\gamma\right)^j \quad \forall i,j\in\BN, i+j\leq k
\end{empheq}
\end{lemma}

\begin{proof}
Consider applying the dual characterization of the approximation ratio (\cref{lem:dual-opt}) with the solution $z = 2\beta/2^k$, $y_{-1}^+ = X\beta/(k2^k)$, $y_0^+ = Y\beta/(k2^k)$, and $y_{+1}^+ = y_{-1}^-=y_0^-=y_{+1}^-=0$; it is sufficient to show that this solution is feasible. Note that $t^+_{-1} = -\delta$ and $t^+_0 = \delta$. Thus, the feasibility constraints in \cref{lem:dual-opt} become
\begin{equation}\label{eq:dual-feas}
    \frac{\beta}{2^k} \left(\1[\vecc \in \PClauses_k^L] \cdot 2 + ((-\delta - 1) c_{-1}^+ + (1-\delta) c_{-1}^-) \frac{X}k + ((\delta-1)c_0^+ + (1+\delta)c_0^-) \frac{Y}k \right) \leq \prob^\vecp(\vecc) \quad \forall \vecc \in \Clauses_k^L.
\end{equation}

By \cref{eq:p}, and since $c_{-1}^++c_0^++c_{+1}^++c_{-1}^-+c_0^-+c_{+1}^-=k$, we have \[ \prob^\vecp(\vecc) = \left(\frac12-\frac\gamma2\right)^{c_{-1}^++c_{+1}^-} \left(\frac12\right)^{c_0^++c_0^-} \left(\frac12+\frac\gamma2\right)^{c_{+1}^++c_{-1}^-} = 2^{-k} (1-\gamma)^{c_{-1}^++c_{+1}^-} (1+\gamma)^{c_{+1}^++c_{-1}^-}. \] Thus, dividing through by $\beta/2^k$, \cref{eq:dual-feas} becomes
\begin{equation}\label{eq:dual-feas-spec}
    \1[\vecc \in \PClauses^L_k] \cdot 2 + ((1-\delta)c_{-1}^- - (1+\delta)c_{-1}^+)\frac{X}k + ((1+\delta)c_0^- - (1-\delta)c_0^+) \frac{Y}k \leq \beta^{-1} (1-\gamma)^{c_{-1}^++c_{+1}^-} (1+\gamma)^{c_{+1}^++c_{-1}^-} \quad \forall \vecc \in \Clauses_k^L.
\end{equation}

Finally, we claim that \cref{eq:dual-feas-spec} is implied by the hypothesis. Indeed, we consider two cases. First, if $\vecc \in \PClauses^L_k$, then $c_{-1}^-=c_0^-=c_{+1}^-=0$ and $c_0^+=k-c_{-1}^+-c_{+1}^+$, so \cref{eq:dual-feas-spec} becomes \[ 2 - (1+\delta)\frac{c_{-1}^+}kX - (1-\delta)\frac{k-c_{-1}^+-c_{+1}^+}kY \leq \beta^{-1} (1-\gamma)^{c_{-1}^+} (1+\gamma)^{c_{+1}^+}. \] This is precisely the second hypothesized inequality, for $c_{-1}^+=i,c_{+1}^+=j$. On the other hand, if $\vecc \not\in \PClauses^L_k$, then we observe that replacing $(c_{-1}^+,c_0^+,c_{+1}^+,c_{-1}^-,c_0^-,c_{+1}^-) \mapsto (0,0,0,c_{-1}^-+c_{-1}^+,c_0^-+c_0^+,c_{+1}^-+c_{+1}^+)$ fixes the RHS of \cref{eq:dual-feas-spec}, while only increasing the LHS; thus, it suffices to prove \cref{eq:dual-feas-spec} only in this extreme case. Hence, we can assume $c_0^- = k-c_{-1}^--c_{+1}^-$, so \cref{eq:dual-feas-spec} becomes \[ (1-\delta) \frac{c_{-1}^-}kX + (1+\delta) \frac{k-c_{-1}^--c_{+1}^-}k Y \leq \beta^{-1} (1-\gamma)^{c_{+1}^-} (1+\gamma)^{c_{-1}^-}, \] which is precisely the first assumed condition for $c_{-1}^-=j,c_{+1}^-=i$.
\end{proof}

\begin{remark}
We chose the specific family of dual solutions used in the proof of \cref{lem:suff-cond} by inspecting an LP solver's output for $k =2$ and $k=3$. Our investigation also suggests that this solution is unique in a certain sense: In the simplest case of $k=2$ and $\delta=\gamma=0$, it appears that \emph{every} optimal feasible solution requires $y_{-1}^+=2/9$, and further, the only solution with only two nonzero $y$ entries sets $y_0^+=1/9$.
\end{remark}

\subsection{Proving \cref{thm:perturb} via analysis of ``slack''}\label{sec:slack-analysis}

Our goal in this subsection is to prove \cref{thm:perturb} by achieving the sufficient conditions outlined in \cref{lem:suff-cond}, modulo some lemmas. Our first step is to show that the conditions in \cref{lem:suff-cond} are fulfilled when $\delta = 0$, $\gamma = \gamma_k$ (see \cref{eq:gamma-k}), $X=2$, $Y=1$, and $\beta = (1-\gamma_k)^{\lfloor k/2\rfloor} (1+\gamma_k)^{\lfloor k/2\rfloor}$. This will imply that the superoblivious algorithm hypothesized in \cref{thm:superobl} indeed yields an $\alpha^*_k$-approximation (via \cref{lem:suff-cond}) and is also the base of the proof of \cref{thm:perturb}. In particular, we show that all but a handful of the conditions in \cref{lem:suff-cond} are \emph{loose} when we plug in this solution.

Note that when $X=2$, $Y=1$, $\delta = 0$, then the LHS of the first inequality in \cref{lem:suff-cond} is $(1-(i+j)/k) + 2j/k = 1 + (j-i)/k$ and the LHS of the second inequality is also $2-(1-(i+j)/k)-2i/k = 1 + (j-i)/k$, and therefore the two inequalities coincide. That is, the hypothesis of \cref{lem:suff-cond} becomes \[ 1+\frac{j-i}k \leq (1-\gamma_k)^{i - \lfloor k/2\rfloor} (1+\gamma_k)^{j+\lfloor k/2\rfloor} \quad \forall i,j\in\BN,i+j \leq k. \] For example, consider the case where $k$ is even, $i = k/2-1$ and $j = k/2+1$; then the LHS is $1+\frac2k$ and the RHS is $(1-\gamma_k)^{-1} (1+\gamma_k) = (1+1/(k+1))/(1-1/(k+1)) = ((k+2)/(k+1))/(k/(k+1)) = (k+2)/k = 1+2/k$, so we have equality. Indeed, we have the following lemma:

\begin{lemma}[``Two-sided Bernoulli inequality'']\label{lem:bern-fancy}
For every $k \geq 2 \in \BN$ and $i,j \in \BN$, $i+j \leq k$, the following holds:

\begin{itemize}
    \item If $k$ is even, then \[ 
    1+\frac{j-i}k \leq \left(1-\frac1{k+1}\right)^{i-k/2}\left(1+\frac1{k+1}\right)^{j-k/2}. \] Further, the inequality is an equality iff $(i,j) \in \{(k/2,k/2),(k/2-1,k/2),(k/2-1,k/2+1)\}$.
    \item If $k$ is odd, then \[ 
    1+\frac{j-i}k \leq \left(1-\frac1k\right)^{i-(k-1)/2}\left(1+\frac1k\right)^{j-(k-1)/2}. \] Further, the inequality is an equality iff $(i,j) \in \{((k-1)/2,(k-1)/2),((k-1)/2,(k+1)/2),((k+1)/2,(k-1)/2)\}$.
\end{itemize}
\end{lemma}

We prove this lemma in \cref{sec:bern} below. Note that regardless of the parity of $k$, there are exactly three tight cases in the inequality, corresponding to six tight cases in the inequalities in \cref{lem:suff-cond} (because there each $(i,j)$ pair generates two inequalities, which coincide if $X=2,Y=1,\delta=0$). Our next lemmas state that there are feasible solutions $(X,Y)$ for these six inequalities when $\delta$ is positive if we can also slightly perturb the right-hand sides:

\begin{lemma}[Strict satisfaction of core inequalities, even case]\label{lem:core-strict-even}
Fix any even $k \geq 2$. There exists $\epsilon_0$ such that for all $0 < \epsilon < \epsilon_0$, the following holds. Let $\beta = (1-1/(k+1))^{k/2}(1+1/(k+1))^{k/2}$ and $\eta = 1 - \beta^{-1} (1-1/(k+1)-\epsilon)^{k/2}(1+1/(k+1)+\epsilon)^{k/2}$. Then there exists $X,Y \geq 0$ satisfying the \emph{strict} inequalities:

\begin{empheq}[left=\empheqlbrace]{align*}
    &\frac12(1-\delta) X < 1 - \eta \\
    &2 - \frac12(1+\delta) X < 1 - \eta \\
    &\left(\frac12+\frac1k\right)(1-\delta)X < (1 - \eta) \left(1+\frac1{k+1}+\epsilon\right)\left(1-\frac1{k+1}-\epsilon\right)^{-1} \\
    &2 - \left(\frac12-\frac1k\right)(1+\delta) X < (1-\eta) \left(1+\frac1{k+1}+\epsilon\right)\left(1-\frac1{k+1}-\epsilon\right)^{-1} \\
    &\frac1k(1+\delta)Y + \frac12(1-\delta)X < (1 - \eta) \left(1-\frac1{k+1}-\epsilon\right)^{-1} \\
    &2 - \frac1k(1-\delta)Y - \left(\frac12-\frac1k\right)(1+\delta)X < (1 - \eta)\left(1-\frac1{k+1}-\epsilon\right)^{-1}.
\end{empheq}
where $\delta = \epsilon$ if $k=2$ and $\delta = 4\eta$ otherwise.
\end{lemma}

\begin{lemma}[Strict satisfaction of core inequalities, odd case]\label{lem:core-strict-odd}
Fix any odd $k \geq 3$. There exists $\epsilon_0 > 0$ such that for all $0 < \epsilon < \epsilon_0$, the following holds. Let $\beta = (1-1/k)^{(k-1)/2}(1+1/k)^{(k-1)/2}$ and $\eta = 1 - \beta^{-1} (1-1/k-\epsilon)^{(k-1)/2}(1+1/k+\epsilon)^{(k-1)/2}$. Then there exists $X,Y \geq 0$ satisfying the \emph{strict} inequalities:

\begin{empheq}[left=\empheqlbrace]{align*}
    &\left(\frac12+\frac1{2k}\right) (1-\delta) X < (1 - \eta) \left(1+\frac1k+\epsilon\right) \\
    &2-\left(\frac12-\frac1{2k}\right) (1+\delta) X < (1 - \eta) \left(1+\frac1k+\epsilon\right) \\
    &\left(\frac12-\frac1{2k}\right) (1-\delta) X < (1 - \eta) \left(1-\frac1k-\epsilon\right) \\
    &2-\left(\frac12+\frac1{2k}\right) (1+\delta) X < (1 - \eta) \left(1-\frac1k-\epsilon\right) \\
    &\frac1k(1+\delta)Y + \left(\frac12-\frac1{2k}\right)(1-\delta)X < 1-\eta \\
    &2-\frac1k(1-\delta)Y-\left(\frac12-\frac1{2k}\right)(1+\delta)X < 1-\eta
\end{empheq}
where $\delta = 5\eta$.
\end{lemma}

We prove these lemmas in \cref{sec:core-strict} below, but for now, we use the lemmas collected in this subsection to prove \cref{thm:perturb}:

\begin{proof}[Proof of \cref{thm:perturb}]
    We consider the case where $k$ is even, and apply \cref{lem:core-strict-even}. (If $k$ were instead odd, we would apply \cref{lem:core-strict-odd}, but the proof would otherwise proceed in the same manner.)
    
    Let $\beta = (1-\frac1{k+1})^{k/2} (1+\frac1{k+1})^{k/2}$, so that our goal is to show that for some $0 \leq \delta \leq 1$ and $0 \leq \epsilon \leq 1-1/k$, the algorithm $\obl$ for $\vect = (\delta,1)$ and $\vecp = (1+\frac1k+\epsilon)$ achieves a ratio strictly better than $2^{-(k-1)} \beta$.
    
    Towards this, suppose we can show that there exist $\delta, \epsilon, X, Y$ such that the following inequalities are all strict:

    \begin{empheq}[left=\empheqlbrace]{align}
    (1+\delta)\left(1-\frac{i+j}k\right)Y - (1-\delta) \frac{j}k X &< \beta^{-1} \left(1-\frac1{k+1}-\epsilon\right)^i\left(1+\frac1{k+1}+\epsilon\right)^j \quad \forall i,j\in\BN, i+j \leq k \label{eq:perturb:suff}. \\
    2 - (1-\delta) \left(1 - \frac{i+j}k\right) Y - (1+\delta)\frac{i}k X&< \beta^{-1} \left(1-\frac1{k+1}-\epsilon\right)^i\left(1+\frac1{k+1}+\epsilon\right)^j \quad \forall i,j\in\BN, i+j \leq k \nonumber
    \end{empheq}

    Since these inequalities are strict, we know exists $\beta' > \beta$ such that they still hold replacing $\beta$ with $\beta'$; therefore, by \cref{lem:suff-cond} the algorithm $\obl$ achieves ratio at least $\beta'$, which strictly exceeds $\beta$.

    Let $\epsilon > 0$ be chosen later; let $\delta=\delta(\epsilon),X=X(\epsilon),Y=Y(\epsilon)$ be the result of applying \cref{lem:core-strict-even}. We claim that $(\epsilon,\delta,X,Y)$ satisfy \cref{eq:perturb:suff} for sufficiently small $\epsilon$. First, we consider the cases $(i,j) \in \{(k/2,k/2),(k/2-1,k/2+1),(k/2-1,k)\}$. Indeed, letting $\eta(\epsilon) = 1 - \beta^{-1}(1-1/(k+1)-\epsilon)^{k/2}(1+1/(k+1)+\epsilon)^{k/2}$. Thus, $1-\eta(\epsilon) = \beta^{-1}(1-1/(k+1)-\epsilon)^{k/2}(1+1/(k+1)+\epsilon)^{k/2}$, which is precisely the RHS of the above inequalities at $i = j=k/2$. Similarly, at $i=k/2-1,j=k/2+1$, the RHS is \[ \beta^{-1} \left(1-\frac1{k+1}-\epsilon\right)^{k/2-1}\left(1+\frac1{k+1}+\epsilon\right)^{k/2+1} = (1-\eta(\epsilon)) \left(1+\frac1{k+1}+\epsilon\right)\left(1-\frac1{k+1}-\epsilon\right)^{-1}, \] and the RHS at $i=k/2-1,j=k/2$ is \[ \beta^{-1}\left(1-\frac1{k+1}-\epsilon\right)^{k/2-1}\left(1+\frac1{k+1}+\epsilon\right)^{k/2} = (1-\eta(\epsilon))\left(1-\frac1{k+1}-\epsilon\right)^{-1}. \] Therefore, we can write \cref{eq:perturb:suff} at $(i,j) \in \{(k/2,k/2),(k/2-1,k/2+1),(k/2-1,k)\}$ equivalently as:

    \begin{empheq}[left=\empheqlbrace]{align}
    &\frac12(1-\delta) X < 1 - \eta \label{eq:perturb:i} \\
    &2 - \frac12(1+\delta) X < 1 - \eta \label{eq:perturb:ii} \\
    &\left(\frac12+\frac1k\right)(1-\delta)X < (1 - \eta) \left(1+\frac1k+\epsilon\right)\left(1-\frac1{k+1}-\epsilon\right)^{-1} \nonumber \\
    &2 - \left(\frac12-\frac1k\right)(1+\delta) X < (1-\eta) \left(1+\frac1{k+1}+\epsilon\right)\left(1-\frac1{k+1}-\epsilon\right)^{-1} \nonumber \\
    &\frac1k(1+\delta)Y + \frac12(1-\delta)X < (1 - \eta) \left(1-\frac1{k+1}-\epsilon\right)^{-1}  \label{eq:perturb:v} \\
    &2 - \frac1k(1-\delta)Y - \left(\frac12-\frac1k\right)(1+\delta)X < (1 - \eta)\left(1-\frac1{k+1}-\epsilon\right)^{-1} \label{eq:perturb:vi}.
    \end{empheq}
    which was precisely the conclusion of \cref{lem:core-strict-even}. 

    Finally, it remains to show that for sufficiently small $\epsilon$, the strict inequalities (\cref{eq:perturb:suff}) also hold for $(i,j) \not\in \{(k/2,k/2),(k/2-1,k/2+1),(k/2-1,k)\}$. For this, we observe that as $\epsilon \to 0$, we have $\delta,\eta \to 0$. Further, we observe that for all $\epsilon$, \cref{eq:perturb:i,eq:perturb:ii} imply \[ \frac{1+\eta}{1+\delta} < \frac{X}2 < \frac{1-\eta}{1-\delta} \], and therefore as $\epsilon \to 0$ we have $X \to 2$. Similarly, \cref{eq:perturb:v,eq:perturb:vi} together imply \[ \frac{2-(1-\eta)(1+\frac1k)-(\frac12-\frac1k)(1+\delta)X}{\frac1k(1-\delta)} < Y < \frac{(1-\eta)(1+\frac1k)-\frac12(1-\delta)X}{\frac1k(1+\delta)}, \] so as $\epsilon \to 0$, $Y \to 1$. Finally, note that in \cref{eq:perturb:suff}, for each $i,j$, the RHS of both inequalities has limit $\beta^{-1}(1-1/(k+1))^i(1+1/(k+1))^j$ as $\epsilon \to 0$, and the LHS of both inequalities has limit $1+(j-i)/k$. By \cref{lem:bern-fancy}, therefore, the limits have a strict inequality as long as $(i,j) \not\in \{(k/2,k/2),(k/2-1,k/2+1),(k/2-1,k)\}$. So for sufficiently small choice of $\epsilon$, $X$ and $Y$ satisfy \cref{eq:perturb:suff}, as desired.
\end{proof}

\subsection{A ``two-sided Bernoulli inequality'': Proving \cref{lem:bern-fancy}}\label{sec:bern}

In this subsection, we prove \cref{lem:bern-fancy}, which (recall) corresponds to the feasibility of the solution $(X,Y) = (2,1)$ in the system of inequalities in \cref{lem:suff-cond} when $\delta=0$ and $\gamma=\gamma_k$. We include the statement of the standard Bernoulli inequality for completeness:

\begin{proposition}[Bernoulli's inequality]\label{prop:bern}
For all $x, r \in \BR$, if $x > -1$ and $r \leq 0$ or $r \geq 1$, then $1+rx \leq (1+x)^r$. Further, the inequality is strict unless $r = 0$ or $r = 1$.
\end{proposition}

\begin{proof}[Proof of \cref{lem:bern-fancy}]
    Firstly, we note that regardless of the parity of $k$, incrementing both $j$ and $i$ fixes the LHS of the desired inequality while strictly decreasing the RHS. That is, for the LHS we have $(j+1)-(i+1)=j-i$, while for the RHS we have \[ \frac{\left(1-\gamma_k\right)^{(i+1)-\lfloor k/2\rfloor}\left(1+\gamma_k\right)^{(j+1)-\lfloor k/2\rfloor}}{\left(1-\gamma_k\right)^{i-\lfloor k/2\rfloor}\left(1+\gamma_k\right)^{j-\lfloor k/2\rfloor}} = \left(1-\gamma_k\right)\left(1+\gamma_k\right) = 1-\gamma_k^2 < 1 \] since $\gamma_k > 0$. Further, we can increment $j$ and $i$ while maintaining the sum at most $k$ iff $j+i \leq k-2$. Therefore, we need to prove the inequality WLOG in the cases $i+j \in \{k-1,k\}$, and further, equality is only possible in these cases.
    
    Now, we proceed with cases based on the parity of $k$. 

    \paragraph{Case: $k$ is even.} In this case, we have $\lfloor k/2\rfloor = k/2$ and (recall) $\gamma_k = 1/(k+1)$. First, we observe that at $j=\frac{k}2,i=\frac{k}2-1$, the LHS and RHS of the desired inequality are both $1+\frac1k$, so we have \emph{equality}. Setting aside this case, since $j,i$ are integers summing to $k-1$ or $k$, we have $j-i \neq 1$. Further, we have the useful equality
    \begin{equation}\label{eq:bern:even-ratio}
        1+\frac2k = \left(\frac{k+2}{k+1}\right) \left(\frac{k+1}k\right) = \left(1+\frac1{k+1}\right) \left(1-\frac1{k+1}\right)^{-1}.
    \end{equation}

    Therefore we have
    \begin{align*}
        1+\frac{j-i}k &= 1+\frac{j-i}2 \cdot \frac2k \\
        &\leq \left(1+\frac2k\right)^{(j-i)/2} \tag{\cref{prop:bern} and $j-i\neq1$} \\
        &= \left(1+\frac1{k+1}\right)^{(j-i)/2} \left(1-\frac1{k+1}\right)^{(i-j)/2} \tag{\cref{eq:bern:even-ratio}} \\
        &= \left(1-\frac1{k+1}\right)^{i-k/2} \left(1+\frac1{k+1}\right)^{j-k/2} \left(1-\frac1{k+1}\right)^{(k-(i+j))/2} \left(1+\frac1{k+1}\right)^{(k-(i+j))/2}
    \end{align*}

    which is only smaller than our desired RHS because $(1-1/(k+1))(1+1/(k+1)) = 1- 1/(k+1)^2 < 1$ and $j+i \leq k$. Finally, we recall that Bernoulli's inequality (\cref{prop:bern}) has equality iff $0$ or $1$; since the exponent is $(j-i)/2$, we have equality in the cases $i=j=\frac{k}2$ and $i=\frac{k}2-1,j=\frac{k}2+1$, respectively.

    \paragraph{Case: $k$ is odd.} In this case, we have $\lfloor k/2 \rfloor = (k-1)/2$ and (recall) $\gamma_k = 1/k$. We are not aware of a comparably slick approach (which only applies Bernoulli's inequality once), so we will have to do slightly more work. We observe that since $i+j\in\{k-1,k\}$, $\lfloor(i+j)/2\rfloor = (k-1)/2$, so we want to show \[ 1+\frac{j-i}k \leq \left(1-\frac1k\right)^{i-\lfloor(i+j)/2\rfloor} \left(1+\frac1k\right)^{j-\lfloor(i+j)/2\rfloor}, \] and also usefully, we have
    \begin{equation}\label{eq:bern:odd-i+j}
        \lfloor(i+j)/2\rfloor + \lceil(i+j)/2\rceil = i+j.
    \end{equation}
    We also use
    \begin{equation}\label{eq:bern:transfer}
        1+\frac1k = \left(1-\frac1{k+1}\right)^{-1}.
    \end{equation}
    
    First, suppose $j \geq i$. Since $j$ is an integer, $j \geq \lceil(i+j)/2\rceil$. Then we have:
    
    \begin{align*}
        1+\frac{j-i}k &= 1+\frac{2j-(i+j)}k \\
        & \leq \left(1+\frac1k\right)^{2j-(i+j)} \tag{\cref{prop:bern}} \\
        &= \left(1+\frac1k\right)^{j-\lfloor(i+j)/2\rfloor}\left(1+\frac1k\right)^{j-\lceil(i+j)/2\rceil} \tag{\cref{eq:bern:odd-i+j}} \\
        &= \left(1+\frac1k\right)^{j-\lfloor(i+j)/2\rfloor}\left(1-\frac1{k+1}\right)^{\lceil(i+j)/2\rceil-j} \tag{\cref{eq:bern:transfer}} \\
        &\leq \left(1+\frac1k\right)^{j-\lfloor(i+j)/2\rfloor}\left(1-\frac1k\right)^{\lceil(i+j)/2\rceil-j} \tag{$1-1/(k+1)>1-1/k$ and $j \geq \lceil (i+j)/2\rceil$} \\
        &= \left(1+\frac1k\right)^{j-\lfloor(i+j)/2\rfloor}\left(1-\frac1k\right)^{i-\lfloor(i+j)/2\rfloor}. \tag{\cref{eq:bern:odd-i+j}}
    \end{align*}
    
    On the other hand, suppose $i > j$. Then we similarly have:

    \begin{align*}
        1+\frac{j-i}k &= 1-\frac{2i-(i+j)}k \\
        & \leq \left(1-\frac1k\right)^{2i-(i+j)} \tag{\cref{prop:bern}} \\
        &= \left(1-\frac1k\right)^{i-\lfloor(i+j)/2\rfloor}\left(1-\frac1k\right)^{i-\lceil(i+j)/2\rceil} \tag{\cref{eq:bern:odd-i+j}} \\
        &\leq \left(1-\frac1k\right)^{i-\lfloor(i+j)/2\rfloor}\left(1-\frac1{k+1}\right)^{i-\lceil(i+j)/2\rceil} \tag{$1-1/(k+1)>1-1/k$ and $i \geq \lceil (i+j)/2\rceil$} \\
        &= \left(1-\frac1k\right)^{i-\lfloor(i+j)/2\rfloor}\left(1+\frac1k\right)^{\lceil(i+j)/2\rceil-i} \tag{\cref{eq:bern:transfer}} \\
        &= \left(1-\frac1k\right)^{i-\lfloor(i+j)/2\rfloor}\left(1+\frac1k\right)^{j-\lfloor(i+j)/2\rfloor}. \tag{\cref{eq:bern:odd-i+j}}
    \end{align*}

    To identify the tight cases: If $j \geq i$, Bernoulli's inequality (\cref{prop:bern}) is only an equality if $j-i \in \{0,1\}$; these correspond to $j=\frac{k-1}2,i=\frac{k-1}2$ and $j=\frac{k+1}2,i=\frac{k-1}2$, respectively. Similarly, if $i > j$ then we need $i-j=1$, corresponding to $j=\frac{k-1}{2},i=\frac{k+1}2$.
\end{proof}

\subsection{An analysis of ``slack'': Proving \cref{lem:core-strict-even,lem:core-strict-odd}}\label{sec:core-strict}

In this section, we prove \cref{lem:core-strict-even,lem:core-strict-odd}, completing the proof of \cref{thm:perturb}.

\begin{proof}[Proof of \cref{lem:core-strict-even}]
    We handle the cases where $k \geq 4$ and $k = 2$ separately. 

    \paragraph{Case: $k = 2$.} 
    We have \[ 1-\eta = \left(\frac43+\epsilon\right)\left(\frac23-\epsilon\right)/\left(\frac43\cdot\frac23\right) = \frac18(4+3\epsilon)(2-3\epsilon),  \] and similarly \[ (1-\eta)\left(\frac43+\epsilon\right)\left(\frac23-\epsilon\right)^{-1} = \frac18\left(4+3\epsilon\right)^2, \] and \[ (1-\eta)\left(\frac23-\epsilon\right)^{-1} = \frac38\left(4+3\epsilon\right). \] Thus, the desired inequalities become, respectively:
    \begin{subequations}
    \begin{empheq}[left=\empheqlbrace]{align}
    &\frac12(1-\delta) X < \frac18(4+3\epsilon)(2-3\epsilon) \label{eq:strict2:i} \\
    &2 - \frac12(1+\delta) X < \frac18(4+3\epsilon)(2-3\epsilon) \label{eq:strict2:ii} \\
    &(1-\delta)X < \frac18\left(4+3\epsilon\right)^2 \label{eq:strict2:iii} \\
    &2 < \frac18\left(4+3\epsilon\right)^2 \label{eq:strict2:iv} \\
    &\frac12(1+\delta)Y + \frac12(1-\delta)X < \frac38\left(4+3\epsilon\right) \label{eq:strict2:v} \\
    &2 - \frac12(1-\delta)Y < \frac38\left(4+3\epsilon\right) \label{eq:strict2:vi}.
    \end{empheq}
    \end{subequations}

    Now, we observe that \cref{eq:strict2:iv} is always satisfied whenever $\epsilon > 0$ (i.e., there is no dependence on $X$ or $Y$). Similarly, we can compare \cref{eq:strict2:iii} to \cref{eq:strict2:i}; multiplying \cref{eq:strict2:i} by $2$ yields $(1-\delta)X < \frac18(4+3\epsilon)(4-6\epsilon)$, which is strictly stronger than \cref{eq:strict2:iii} whenever $\epsilon > 0$. Thus, we may restrict our attention to the four inequalities \cref{eq:strict2:i,eq:strict2:ii,eq:strict2:v,eq:strict2:vi}.

    We rewrite \cref{eq:strict2:v,eq:strict2:vi} as \[ Y < \frac{\frac38(4+3\epsilon) - \frac12(1-\delta)X}{\frac12(1+\delta)} \text{ and } Y > \frac{2-\frac38(4+3\epsilon)}{\frac12(1-\delta)}, \] respectively. Thus, $Y$ exists iff \[ \frac{2-\frac38(4+3\epsilon)}{\frac12(1-\delta)} < \frac{\frac38(4+3\epsilon) - \frac12(1-\delta)X}{\frac12(1+\delta)} \] which cross-multiplies to \[ \left(\frac38(4+3\epsilon) - \frac12(1-\delta)X\right)\left(\frac12(1-\delta)\right) - \left(2-\frac38(4+3\epsilon)\right)\left(\frac12(1+\delta)\right) > 0, \] and this in turn simplifies to
    \begin{equation}\label{eq:core-strict-2:Y-cond}
        X < \frac{\left(\frac38(4+3\epsilon)\right)\left(\frac12(1-\delta)\right) - \left(2-\frac38(4+3\epsilon)\right)\left(\frac12(1+\delta)\right)}{\left(\frac12(1-\delta)\right)^2} = \frac{4-8\delta+9\epsilon}{2(1-\delta)^2}.
    \end{equation}

    Now a solution exists iff the upper bounds in \cref{eq:core-strict-2:Y-cond,eq:strict2:i} are both compatible with the lower bound in \cref{eq:strict2:ii}. We can rewrite \cref{eq:strict2:i,eq:strict2:ii} as \[ X < \frac{\frac18(4+3\epsilon)(2-3\epsilon)}{\frac12(1-\delta)} = \frac{8-6\epsilon-9\epsilon^2}{4(1-\delta)} \text{ and } X > \frac{2-\frac18(4+3\epsilon)(2-3\epsilon)}{\frac12(1+\delta)} = \frac{8+6\epsilon+9\epsilon^2}{4(1+\delta)}, \] respectively. Finally, we plug in $\delta = \epsilon$, and for $0 < \epsilon < \epsilon_0 := 2/9$, we have \[ \frac{8+6\epsilon+9\epsilon^2}{4(1+\epsilon)} < \frac{8-6\epsilon-9\epsilon^2}{4(1-\epsilon)} < \frac{4-\epsilon}{2(1-\epsilon)^2}, \] which can be verified by again cross-multiplying, and therefore $X$ exists.

    \paragraph{Case: $k \geq 4$.} In this case, we claim that the following set of inequalities, which correspond to substituting $\epsilon = 0$ only on the RHS, are stronger than the desired inequalities:

    \begin{subequations}
    \begin{empheq}[left=\empheqlbrace]{align}
        &\frac12(1-\delta) X < 1 - \eta \label{eq:strict4:i} \\
        &2 - \frac12(1+\delta) X < 1 - \eta \label{eq:strict4:ii} \\
        &\left(\frac12+\frac1k\right)(1-\delta)X < (1 - \eta) \left(1+\frac2k\right) \label{eq:strict4:iii} \\
        &2 - \left(\frac12-\frac1k\right)(1+\delta) X < (1-\eta) \left(1+\frac2k\right) \label{eq:strict4:iv} \\
        &\frac1k(1+\delta)Y + \frac12(1-\delta)X < (1 - \eta) \left(1+\frac1k\right) \label{eq:strict4:v} \\
        &2 - \frac1k(1-\delta)Y - \left(\frac12-\frac1k\right)(1+\delta)X < (1 - \eta)\left(1+\frac1k\right) \label{eq:strict4:vi}.
    \end{empheq}
    \end{subequations}

    Indeed, the first two inequalities are the same. For the latter four, we have $(1-1/(k+1)-\epsilon)^{-1} > (1-1/(k+1))^{-1} = 1+1/k$ and $(1+1/(k+1)+\epsilon)(1-1/(k+1)-\epsilon)^{-1} > (1+1/(k+1))(1-1/(k+1))^{-1}=1+2/k$.\footnote{The former inequality follows since inversion reverses the direction of $1-1/(k+1)-\epsilon < 1-1/(k+1)$. For the latter, cross-multiply to get $(1+1/(k+1)+\epsilon)(1-1/(k+1))-(1+1/(k+1))(1-1/(k+1)-\epsilon)=2\epsilon > 0$.} Now, we observe that \cref{eq:strict4:iii} is precisely \cref{eq:strict4:i} multiplied on both sides by $1+\frac2k$, so we can ignore \cref{eq:strict4:iii} and focus on the remaining equations. We can rewrite \cref{eq:strict4:i} equivalently as 
    \begin{equation}\label{eq:strict4:X-ub-1}
        X < \frac{2(1-\eta)}{1-\delta}.
    \end{equation}
    Next, we can rewrite \cref{eq:strict4:ii} as \[ X > \frac{2(1+\eta)}{1+\delta} \] and \cref{eq:strict4:iv} as
    \begin{equation}\label{eq:strict4:X-lb}
        X > \frac{2-(1-\eta)\left(1+\frac2k\right)}{(1+\delta)\left(\frac12-\frac1k\right)} = \frac{2(1+\eta\frac{k+2}{k-2})}{1+\delta}.
    \end{equation}
    Thus, \cref{eq:strict4:iv} is strictly stronger than \cref{eq:strict4:ii}, so we can similarly ignore \cref{eq:strict4:ii}.\footnote{Note the difference with the case $k=2$, where the fourth equation was \emph{weaker} than the second (and in fact held tautologically).}

    Next, we rewrite \cref{eq:strict4:v,eq:strict4:vi} as \[ Y < \frac{(1-\eta)(1+\frac1k)-\frac12(1-\delta)X}{\frac1k(1+\delta)} \text{ and } Y > \frac{2-(1-\eta)(1+\frac1k)-(\frac12-\frac1k)(1+\delta)X}{\frac1k(1-\delta)}. \] Thus, $Y$ exists if and only if $X$ satisfies \[ (1-\delta)\left((1-\eta)\left(1+\frac1k\right)-\frac12(1-\delta)X\right)-(1+\delta)\left(2-(1-\eta)\left(1+\frac1k\right)-\left(\frac12-\frac1k\right)(1+\delta)X\right) > 0, \] which simplifies to
    \[
        X < \frac{2(1-k\delta - (k+1)\eta)}{1-2(k-1)\delta+\delta^2}.
    \] Setting $\delta = 4\eta$, and using the assumption $\delta \leq 1$, this is implied by
    \begin{equation}\label{eq:strict4:X-ub-2}
        X < \frac{2(1-(5k+1)\eta)}{1-(8k-9)\eta}.
    \end{equation}

    Finally, we observe that $X$ exists if the lower bound in \cref{eq:strict4:X-lb} is compatible with the upper bounds in \cref{eq:strict4:X-ub-1,eq:strict4:X-ub-2}. For $0 < \eta < \eta_0 := \min\{1/(8k-9),(13-3k)/(13+12k)\}$, we have \[ \frac{2(1+\eta)}{1+4\eta} < \frac{2(1-(5k+1)\eta)}{1-(8k-9)\eta} <  \frac{2(1-\eta)}{1-4\eta} \] which can be verified by cross-multiplying, and therefore $X$ exists.
\end{proof}

\begin{proof}[Proof of \cref{lem:core-strict-odd}]
Substituting $Z=\frac12X$, we rewrite the six inequalities as, respectively:
    
\begin{subequations}
\begin{empheq}[left=\empheqlbrace]{align}
    & Z < \frac{(1 - \eta) \left(1+\frac1k+\epsilon\right)}{\left(1+\frac1k\right) (1-\delta)}\label{eq:core-strict-odd:i} \\
    & Z > \frac{2- (1 - \eta) \left(1+\frac1k+\epsilon\right)}{\left(1-\frac1k\right) (1+\delta)}\label{eq:core-strict-odd:ii}\\
    & Z < \frac{(1 - \eta) \left(1-\frac1k-\epsilon\right)}{\left(1-\frac1k\right) (1-\delta)} \label{eq:core-strict-odd:iii}\\
    & Z > \frac{2-(1 - \eta) \left(1-\frac1k-\epsilon\right)}{\left(1+\frac1k\right) (1+\delta)} \label{eq:core-strict-odd:iv}\\
    & Y < \frac{1-\eta-\left(1-\frac1{k}\right)(1-\delta)Z}{\frac1k(1+\delta)} \label{eq:core-strict-odd:v}\\
    & Y > \frac{1+\eta-\left(1-\frac1{k}\right)(1+\delta)Z}{\frac1k(1-\delta)}\label{eq:core-strict-odd:vi}
\end{empheq}
\end{subequations}

We compare \cref{eq:core-strict-odd:v,eq:core-strict-odd:vi}, and we see that $Y$ exists iff \[ \frac{1+\eta-\left(1-\frac1{k}\right)(1+\delta)Y}{1-\delta} < \frac{1-\eta-\left(1-\frac1{k}\right)(1-\delta)Y}{1+\delta}. \] We can cross-multiply, and deduce that the $Y$ exists iff
\begin{equation}\label{eq:core-strict-odd:Y-cond}
    Z > \frac{\eta+\delta}{2\delta\left(1-\frac1k\right)}
\end{equation}

Now, we compare \cref{eq:core-strict-odd:i,eq:core-strict-odd:iii}; we have \[ \frac{1-\frac1k-\epsilon}{1-\frac1k} < \frac{1+\frac1k+\epsilon}{1+\frac1k} \] since the inequality cross-multiplies to $2\epsilon>0$, and therefore \cref{eq:core-strict-odd:iii} implies \cref{eq:core-strict-odd:i}.

Finally, we observe that $Z$ exists iff the upper-bound on $Z$ in \cref{eq:core-strict-odd:iii} exceeds the lower-bounds in \cref{eq:core-strict-odd:ii,eq:core-strict-odd:iv,eq:core-strict-odd:Y-cond}. We prove each inequality separately. First, \cref{eq:core-strict-odd:iii,eq:core-strict-odd:Y-cond} are compatible iff \[ \frac{(1-\eta)(1-\frac1k-\epsilon)}{1-\delta} > \frac{\delta+\eta}{2\delta}. \] At $\delta = 5\eta$ this cross-multiplies to $2k-5(1-\eta)-5\epsilon k(1-\eta) +10k\eta > 0$, which since $k \geq 3$ holds for sufficiently small $\epsilon$. Second, \cref{eq:core-strict-odd:iii,eq:core-strict-odd:ii} are compatible iff \[ \frac{(1-\eta)(1-\frac1k-\epsilon)}{1-\delta} > \frac{2-(1-\eta)(1+\frac1k+\epsilon)}{1+\delta}, \] and at $\delta = 5\eta$ this cross-multiplies to $4k-5(1-\eta)-5k\epsilon(1-\eta) > 0$, which again holds for sufficiently small $\epsilon$ since $k \geq 3$. Finally, \cref{eq:core-strict-odd:iii,eq:core-strict-odd:iv} are compatible iff \[ \frac{(1-\eta)(1-\frac1k-\epsilon)}{\left(1-\frac1k\right)(1-\delta)} > \frac{2-(1-\eta)(1-\frac1k-\epsilon)}{\left(1+\frac1k\right)(1+\delta)}; \] at $\delta = 5\eta$ this cross-multiplies to $(4k+5-5\eta)(k-1)\eta - \epsilon k (1-\eta)(k+5\eta) > 0$. Sadly, to complete the proof we will actually need to compare $\eta$ and $\epsilon$. Suppose the following claim:

\begin{claim*}
For sufficiently small $\epsilon > 0$, we have $\epsilon/2 \leq \eta \leq \epsilon$.
\end{claim*}

Then we have the lower-bound $(4k+5-5\eta)(k-1)\eta - \epsilon k (1-\eta)(k+5\eta) \geq (4k+5)(k-1)\epsilon/2 - \epsilon k(k+5\epsilon)$, which is positive for sufficiently small $\epsilon$ because $k \geq 3$ so $(4k+5)(k-1) \geq k^2$. 

It remains to check the claim. We expand \[ \eta = 1-\beta^{-1} \left(1-\frac1k-\epsilon\right)^{(k-1)/2} \left(1+\frac1k+\epsilon\right)^{(k-1)/2} = 1-\left(1-\frac{\epsilon}{1-\frac1k}\right)^{(k-1)/2}\left(1+\frac{\epsilon}{1+\frac1k}\right)^{(k-1)/2}. \] When we expand the expression on the RHS, the $1$'s cancel, and the coefficient of $\epsilon$ is $-(\frac{k-1}2)(-\frac1{1-1/k} + \frac1{1+1/k}) = \frac{k}{k+1}$. (We have $(k-1)/2$ choices of the $-\epsilon/(1-1/k)$ to pick from the left factor, taking $1$ from all remaining factors, and similarly for the right factor.) All remaining terms are lower order in $\epsilon$, and $1/2 < k/(k+1) < 1$, so for sufficiently small $\epsilon > 0$, we have $\epsilon/2 \leq \eta \leq \epsilon$, as desired.
\end{proof}

\section{The limitations of ``superoblivious'' algorithms: Proving \cref{thm:superobl}}\label{sec:thm:superobl}

Recall, a ``superoblivious'' algorithm for $\mkand$ is one of the form $\obl$ for $\vect = (0,1)$, $\vecp = (p)$, i.e., the algorithm assigns positively-biased vertices to $1$ w.p. $p$, negatively-biased vertices to $1$ w.p. $1-p$, and zero-bias vertices to $1$ w.p. $\frac12$. In this section, we prove \cref{thm:superobl}, which states that the best superoblivious algorithm achieves an approximation ratio of exactly $\alpha^*_k$ on the $\mkand$ problem. The main thrust of the section is proving the following lemma:

\begin{lemma}\label{lem:primal-two}
    For all $k \geq 2 \in \BN$, let $\vect = (0,1)$ (and $L=3$). There exists a feasible solution $\{W(\vecc)\}_{\vecc \in \Clauses^L_k}$ to the primal linear program in \cref{lem:primal-opt}, such that for all $p \in [0,1]$, the value of the objective function is upper-bounded: \[ \sum_{\vecc \in \Clauses^L_k} \prob^{(p)}(\vecc) \cdot W(\vecc) \leq \alpha^*_k. \] Further, there is equality iff $p = p^*_k$.
\end{lemma}

Such a feasible solution can be viewed alternatively as a hard instance of $\mkand$ for superoblivious algorithms. Indeed, given \cref{lem:primal-two} and its matching ``dual'' construction in \cref{lem:suff-cond,lem:bern-fancy}, we can immediately prove \cref{thm:superobl}:

\begin{proof}[Proof of \cref{thm:superobl}]
    First, we claim that if $p \neq p^*_k$, then $\alpha(\obl) < \alpha^*_k$. For this, we combine the (strict) upper bound in \cref{lem:primal-two} with \cref{lem:primal-opt}. On the other hand, we claim that if $p = p^*_k$, then $\alpha(\obl) = \alpha^*_k$. That $\alpha(\obl) \leq \alpha^*_k$ follows by the same pair of lemmas, while that $\alpha(\obl) \geq \alpha^*_k$ follows from \cref{lem:suff-cond} and \cref{lem:bern-fancy}. (See the discussion in \cref{sec:slack-analysis} for an explanation of why the guarantee of \cref{lem:bern-fancy} is precisely the hypothesis of \cref{lem:suff-cond} for $\delta=0,X=2,Y=1$.)
\end{proof}

To prove \cref{lem:primal-two}, we construct the feasible solution by ``synthesizing'' certain pairs of instances, developed in \cite{BHP+22}, which are indistinguishable for \emph{sketching} algorithms. It is not clear how to do this in a black-box way, so we have opted to present the lower bound from the ground-up while reusing some key inequalities which also arose in the lower bound of \cite{BHP+22}. We discuss the connection between ``oblivious-hard'' instances and ``sketching-hard'' pairs of instances again at the end of this section. But for now, from \cite{BHP+22}, we re-use the following inequalities:

\begin{proposition}[{\cite[from proof of Lemma 17]{BHP+22}}]\label{prop:bhp}
For all $k \geq 2 \in \BN$, the following holds.
    \begin{itemize}
        \item For odd $k$, the polynomial $r_k(p) := p^{(k+1)/2}(1-p)^{(k-1)/2}$ is uniquely minimized over the unit interval at $p = p^*_k$. In particular, \[ \frac{\min_{p\in[0,1]} r_k(p)}{\frac12\left(1+\frac1k\right)} = \frac{r_k(p^*_k)}{\frac12\left(1+\frac1k\right)} = \alpha^*_k. \]
        \item For even $k$, the polynomial $r_k(p) := \frac{(k+2)^2}{k^2+(k+2)^2} p^{k/2+1}(1-p)^{k/2-1} + \frac{k^2}{k^2+(k+2)^2} p^{k/2}(1-p)^{k/2}$ is uniquely minimized over the unit interval at $p = p^*_k$. In particular, \[ \frac{\min_{p\in[0,1]} r_k(p)}{\frac{(k+1)(k+2)}{k^2+(k+2)^2}} = \frac{r_k(p^*_k)}{\frac{(k+1)(k+2)}{k^2+(k+2)^2}} = \alpha^*_k. \]
    \end{itemize}
\end{proposition}

For completeness, we note that this proposition can be proved by noting that the derivative $\frac{d r_k}{dp}$ equals \[ -(1-p)^{(k-3)/2}p^{(k-1)/2}\left(kp-\frac{k+1}2\right) \] and \[ -\frac{k}{2+2k+2k^2}(1-p)^{k/2-2} p^{k/2-1} \left(\frac{k}2+1-2p\right)\left((k+1)p-\left(\frac{k}2+1\right)\right) \] in the cases of odd and even $k$, respectively, and $p^*_k$ is therefore the only critical point of $r^*_k$ in the unit interval. (Note that for odd $k$, $kp^*_k-(k+1)/2=0$, and for even $k$, $(k+1)p^*_k-(k/2+1)=0$.)

Now, we are prepared to give a proof of \cref{lem:primal-two}:

\begin{proof}[Proof of \cref{lem:primal-two}]
We construct a \emph{sparse} solution $W$. Since $L = 3$ (i.e., there are three bias classes $[1,0),\{0\},(0,1]$), $W$ is indexed by sextuples $(c_{-1}^+,c_0^+,c_{+1}^+,c_{-1}^-,c_0^-,c_{+1}^-)$ of natural numbers which sum to $k$. We split into cases based on the parity of $k$.

\paragraph{Case: $k$ odd.} Let $\veca = (\frac{k-1}2,0,\frac{k+1}2,0,0,0)$ and let $\vecb=(0,0,0,\frac{k+1}2,0,\frac{k-1}2)$. (I.e., $\veca$ corresponds to clauses with all positive literals, and a bare majority of positively-biased variables. $\vecb$ corresponds to clauses with all negative literals, and a bare majority of negatively-biased variables.) For convenience, define $\gamma := \frac12(1+\frac1k)$. Then we set $W(\veca) := \gamma^{-1} \cdot \frac12(1+\frac1k)$, $W(\vecb) := \gamma^{-1} \cdot \frac12(1-\frac1k)$, and $W(\vecc) = 0$ for all $\vecc \in \Clauses^L_k \setminus \{\veca,\vecb\}$. (So, $W(\veca) = 1$, and $W(\vecb) = 1-\frac2{k+1}$. We write $W$ this way so as to to ``change'' the normalization by a factor of $\gamma$, letting us think of the \emph{total} weight in $W$ as $1$.)

First, we check that $W$ is feasible. Noting that $\veca \in \PClauses^L_k$ and $\vecb \not\in \PClauses^L_k$, we have $\sum_{\vecc \in \PClauses^L_k} W(\vecc) = W(\veca) = 1$. We also need to verify the inequalities \[ t_{-1}^-(W^+(-1) + W^-(-1)) \leq W^+(-1) - W^-(-1) \leq t_{-1}^+(W^+(-1) + W^-(-1)), \]
\[ t_0^-(W^+(0) + W^-(0)) \leq W^+(0) - W^-(0) \leq t_0^+(W^+(0) + W^-(0)), \] and
\[ t_{+1}^-(W^+(+1) + W^-(+1)) \leq W^+(+1) - W^-(+1) \leq t_{+1}^+(W^+(+1) + W^-(+1)). \] Recalling that $t^-_{-1} = -1$, $t^+_{-1} = t^-_0 = t^+_0 = t^-_{+1} = 0$, and $t^+_{+1} = +1$, these are equivalent to, respectively, \[ W^+(-1) \leq W^-(-1), W^+(0) = W^-(0),\text{ and } W^-(+1) \leq W^+(+1). \]

Now recall by definition of $W^+(i)$ and $W^-(i)$ (in \cref{lem:primal-opt}), we have

\begin{align*}
    W^+(+1) &= a^+_{+1} W(\veca) + b^+_{+1} W(\vecb) \\
    W^-(+1) &= a^-_{+1} W(\veca) + b^-_{+1} W(\vecb)  \\
    W^+(0) &= a^+_0 W(\veca) + b^+_0 W(\vecb) \\
    W^-(0) &= a^-_0 W(\veca) + b^-_0 W(\vecb) \\
    W^+(-1) &= a^+_{-1} W(\veca) + b^+_{-1} W(\vecb) \\
    W^-(-1) &= a^-_{-1} W(\veca) + b^-_{-1} W(\vecb).
\end{align*}

Luckily, $a^-_{+1},b^+_{+1},a^+_0,b^+_0,a^-_0,b^-_0,a^-_{-1},b^+_{-1}$ are all zero by definition, so we conclude $W^+(0) = W^-(0) = 0$ and 
\begin{alignat*}{4}
    W^+(+1) &= a^+_{+1} W(\veca) &&= \frac1\gamma \left(\frac{k+1}2\right) \left(\frac12\left(1+\frac1k\right)\right) &&= \frac1{\gamma \cdot 4k} (k+1)^2 \\
    W^-(+1) &= b^-_{+1} W(\vecb) &&= \frac1\gamma \left(\frac{k-1}2\right) \left(\frac12\left(1-\frac1k\right)\right) &&= \frac1{\gamma \cdot 4k} (k-1)^2 \\
    W^+(-1) &= a^+_{-1} W(\veca) &&= \frac1\gamma \left(\frac{k-1}2\right) \left(\frac12\left(1+\frac1k\right)\right) &&= \frac1{\gamma \cdot 4k} (k+1)(k-1) \\
    W^-(-1) &= b^-_{-1} W(\vecb) &&= \frac1\gamma \left(\frac{k+1}2\right) \left(\frac12\left(1-\frac1k\right)\right) &&= \frac1{\gamma \cdot 4k} (k+1)(k-1)
\end{alignat*}

Since $(k-1)^2 \leq (k+1)(k-1) \leq (k+1)^2$, therefore, we satisfy both of the desired inequalities, and so $W$ is feasible.

Finally, we prove the upper bound on $W$'s objective value. We recall from \cref{eq:p} that $\prob^{(p)}(\veca) = \prob^{(p)}(\vecb) = p^{(k+1)/2}(1-p)^{(k-1)/2}$. Therefore, the objective function has value \[ \prob^{(p)}(\veca)\cdot W(\veca) + \prob^{(p)}(\vecb)\cdot W(\vecb) = p^{(k+1)/2}(1-p)^{(k-1)/2} (W(\veca)+W(\vecb)) = \frac{p^{(k+1)/2}(1-p)^{(k-1)/2}}{\frac12(1+\frac1k)} \] by the definition of $\gamma$, which is precisely $r_k(p)/(\frac12(1+\frac1k))$. Thus, by \cref{prop:bhp}, we get that the objective function is upper-bounded by $\alpha^*_k$, with equality precisely when $p = p^*_k$.

\paragraph{Case: $k$ even.} Let $\veca = (\frac{k}2, 0, \frac{k}2, 0, 0, 0)$, let $\vecb = (\frac{k}2+1,0,\frac{k}2-1,0,0,0)$, and let $\vecd = (0,0,0,\frac{k}2,0,\frac{k}2)$. (I.e., $\veca$ corresponds to clauses with all positive literals, balanced between positively- and negatively-biased variables. $\vecb$ corresponds to clauses with all positive literals, with a bare majority of positively-biased variables. $\vecd$ corresponds to clauses with all negative literals, balanced between positively- and negatively-biased variables.) We let $\gamma = \frac{(k+1)(k+2)}{k^2+(k+2)^2}$. Then, we define $W(\veca) = \gamma^{-1} \frac{3k+2}{k^2+(k+2)^2}$; $W(\vecb) = \gamma^{-1} \frac{k^2}{k^2+(k+2)^2}$; $W(\vecd) = \gamma^{-1} \frac{k^2+k+2}{k^2+(k+2)^2}$; and $W(\vecc) = 0$ for all $\vecc \in \Clauses^L_k \setminus \{\veca,\vecb,\vecd\}$.

We proceed with a similar but slightly messier analysis to the previous case. Again, we have by definition that $\sum_{\vecc \in \PClauses^L_k} W(\vecc) = W(\veca) + W(\vecb) = 1$ by definition of $W$ and $\gamma$. We again have $W^+(0) = W^-(0) = 0$. Further, by a similar calculation,

\begin{alignat*}{4}
    W^+(+1) &= a^+_{+1} W(\veca) + b^+_{+1} W(\vecb) &&= \frac1{\gamma\cdot(k^2+(k+2)^2)} \left(\left(\frac{k}2\right) (3k+2) + \left(\frac{k}2+1\right) k^2\right) \\
    W^-(+1) &= d^-_{+1} W(\vecd) &&= \frac1{\gamma\cdot(k^2+(k+2)^2)} \left(\frac{k}2\right) (k^2+k+2) \\
    W^+(-1) &= a^+_{-1} W(\veca) + b^+_{-1} W(\vecb) &&= \frac1{\gamma\cdot(k^2+(k+2)^2)} \left(\left(\frac{k}2\right) (3k+2) + \left(\frac{k}2-1\right) k^2\right) \\
    W^-(-1) &= d^-_{-1} W(\vecd) &&= \frac1{\gamma\cdot(k^2+(k+2)^2)} \left(\frac{k}2\right) (k^2+k+2).
\end{alignat*}

Now we can verify that \[ \left(\frac{k}2\right) (k^2+k+2) = \left(\frac{k}2\right) (3k+2) + \left(\frac{k}2-1\right) k^2 < \left(\frac{k}2\right) (3k+2) + \left(\frac{k}2+1\right) k^2 \] and therefore $W$ is feasible.

To conclude, we recall from \cref{eq:p} that $\prob^{(p)}(\veca) = \prob^{(p)}(\vecd) = p^{k/2}(1-p)^{k/2}$ while $\prob^{(p)}(\vecb) = p^{k/2+1}(1-p)^{k/2-1}$. Thus, the objective function has value \[ \prob^{(p)}(\veca)\cdot W(\veca) + \prob^{(p)}(\vecb)\cdot W(\vecb) + \prob^{(p)}(\vecd)\cdot W(\vecd) =  \frac{k^2 p^{k/2}(1-p)^{k/2} + (k+2)^2 p^{k/2+1}(1-p)^{k/2-1}}{(k+1)(k+2)}. \] This is, again, precisely $r_k(p)/((k+1)(k+2)/(k^2+(k+2)^2))$, so by \cref{prop:bhp}, we get that the objective function is upper-bounded by $\alpha^*_k$, with equality when $p = p^*_k$.
\end{proof}

\section{Implications for streaming algorithms}\label{sec:streaming}

In this section, we prove \cref{cor:random-order,cor:bounded-degree}, translating the oblivious algorithms we developed in the previous sections into streaming algorithms; specifically, we will develop algorithms that work under two different assumptions, namely, random-ordering and bounded-degree. This translation follows the analysis of Saxena, Singer, Sudan, and Velusamy~\cite{SSSV23-random-ordering} for the $\mdcut$ problem.

For any instance $\Psi$ of $\mkand$, and any $L$-class bias partition $\vect$, we define a so-called \emph{snapshot array} for $\Psi$ which captures the (relative) weights of constraints with each possible pattern, generalizing the definition for the case of $\mdcut$ in \cite{SSSV23-random-ordering,SSSV23-dicut}. This is an array $\Snap \in \BR_{\geq 0}^{\Clauses_k^L}$ given by \[ \Snap(\vecc) \eqdef \frac{\sum_{j=1}^m w_j \1[\ptn(C_j) = \vecc]}{\sum_{j=1}^m w_j} \] when $\Psi$ has clauses $C_1,\ldots,C_m$ with weights $w_1,\ldots,w_m$. Note that \cref{claim:obl-exp} above states that $\obl(\Psi)$, the expected value of the assignment produced by an oblivious algorithm on the instance $\Psi$, is precisely a linear combination of the entries of $\Snap$, where the weights depend only on the rounding vector $\vecp$ (and in particular, the weight on entry $\vecc$ is $\prob^\vecp(\vecc)$). As a corollary, we have the following:

\begin{proposition}\label{prop:mhat-suff}
    For every $L, k \geq 2 \in \BN$, suppose $\vect$ is a bias partition and $\vecp$ a rounding vector such that $\obl$ achieves a ratio $\alpha(\obl) \geq \alpha$. For every instance $\Psi$ of $\mkand$, suppose $\hat{M} \in \BR^{\Clauses_k^L}$ is an estimate for $\Snap$ in the sense that $\|\Snap-\hat{M}\|_1 := \sum_{\vecc \in \Clauses^L_k} |\Snap(\vecc) - \hat{M}(\vecc)| \leq \epsilon$. Then \[ (\alpha - 2^{k+1} \epsilon) \val_\Psi \leq \sum_{\vecc \in \Clauses_k^L} \prob^\vecp(\vecc) \hat{M}(\vecc) - \epsilon \leq \val_\Psi. \]
\end{proposition}

\begin{proof}
    Use the assumption $\alpha \val_\Psi \leq \sum_{\vecc \in \Clauses_k^L} \prob^\vecp(\vecc) M(\vecc) \leq \val_\Psi$ together with \[ \left\lvert\sum_{\vecc \in \Clauses_k^L} \prob^\vecp(\vecc) \hat{M}(\vecc) - \sum_{\vecc \in \Clauses_k^L} \prob^\vecp(\vecc) M(\vecc)\right\rvert \leq \sum_{\vecc \in \Clauses_k^L} \prob^\vecp(\vecc) |\hat{M}(\vecc) - M(\vecc)| \leq \sum_{\vecc \in \Clauses_k^L} |\hat{M}(\vecc) - M(\vecc)| \leq \epsilon \] and $\val_\Psi \geq 2^{-k}$.
\end{proof}

Both of the algorithms we describe (to prove \cref{cor:bounded-degree,cor:random-order}) have the following structure: First, we fix the bias partition $\vect$ and rounding vector $\vecp$ coming from \cref{thm:perturb} which yield the better-than-$\alpha^*_k$ approximation. Now, given an instance $\Psi$, we claim that is sufficient to produce an array $\hat{M} \in \BR^{\Clauses_k^L}$ such that for all $\vecc \in \Clauses_k^L$ we have $|\hat{M}(\vecc)-\Snap(\vecc)| \leq \epsilon$. Indeed, by the triangle inequality we will have $\|\hat{M}-\Snap\|_1 \leq K\epsilon$, where $K := |\Clauses_k^L|$, and then we can apply the above proposition (\cref{prop:mhat-suff}), which implies that we can produce a $(\alpha - 2^{k+1} K\epsilon)$-approximation to $\val_\Psi$ (via a linear function of $\hat{M}$'s entries); reparametrizing $\epsilon$ to drop these factors will yield the final algorithm.

Now first, we develop the random-ordering algorithm:

\begin{proof}[Proof of \cref{cor:random-order}]
Let $p:=\epsilon^2/C$ where $C = C(k)$ is a large constant to be chosen later. We store a set $E$ containing the first $1/p$ constraints in the stream. Letting $S$ denote the set of variables appearing in the constraints in $E$, over the remainder of the stream, we track the bias of every variable in $S$. Finally, at the end of the stream, we estimate $\Snap$ via $\hat{M}(\vecc) = \frac1{pm} \hat{X}_\vecc$ where $\hat{X}_\vecc$ is the number of constraints in $E$ with pattern $\vecc$. This method clearly runs using only $O(\log n/\epsilon^2)$ space, since we need only to store $E$ and the bias of every vertex in $S$; both of these sets have constant size $O(1/\epsilon^2)$. So, by the reasoning in the above paragraph, it suffices to show that for every $\vecc \in \Clauses_k^L$, $|\hat{M}(\vecc) - \Snap(\vecc)| \leq \epsilon$.

For each pattern $\vecc$, let $X_\vecc := \Snap(\vecc) \cdot m$ denote the true number of constraints in $\Psi$ with pattern $\vecc$. So by definition of $\hat{M}$ and $X_\vecc$, we equivalently want to show $|\hat{X}_\vecc - p X_\vecc| \leq \epsilon p m$. Now note that $\hat{X}_\vecc$ equals the number of constraints with pattern $\vecc$ in a random sample of $1/p$ constraints in $\Psi$ drawn \emph{without replacement}. We argue that instead, we could prove the inequality for $\tilde{X}_\vecc$ which equals the number of constraints with pattern $\vecc$ in a random sample of $1/p$ constraints in $\Psi$ drawn \emph{with} replacement. Indeed, when we sample $\tilde{X}_\vecc$, let $\CE$ denote the event that we sample no constraint $C_j, j\in[m]$ twice. Then the distribution of $\hat{X}_\vecc$ is the same as the distribution of $\tilde{X}_\vecc$ conditioned on $\overline{\CE}$; $\CE$ has probability $o(1)$,\footnote{Technical note: This holds by Markov's inequality, but only if there are $\omega(1)$ constraints and therefore $o(1)$ expected collisions. But if not, we can simply store all the constraints in the stream.} and therefore we can freely discard all samples where $\CE$ occurs. Finally, we observe that $\tilde{X}_\vecc$ can be written as the sum of $1/p$ independent Bernoulli variables, each taking $1$ w.p. $X_\vecc/m$. Therefore
\begin{equation}\label{eq:random-order:Xhat-exp}
    \Exp[\tilde{X}_\vecc] = X_\vecc/(pm)
\end{equation}
and using the Chernoff bound:
\begin{align*}
    \Pr[|\tilde{X}_\vecc - p X_\vecc| \leq \epsilon p m] &\leq 2\exp\left(\frac{-(\epsilon p m)^2}{3\Exp[\tilde{X}_\vecc]}\right) \tag{Chernoff bound} \\
    &\leq 2\exp\left(\frac{-\epsilon^2 p m^2}{3 X_\vecc}\right) \tag{\cref{eq:random-order:Xhat-exp}} \\
    &\leq 2\exp\left(\frac{-\epsilon^2 p m}{3}\right) \tag{$X_\vecc \leq m$} \\
    &\leq 2\exp(-C/3). \tag{def. of $p$}
\end{align*}
Setting $C=C(k)$ sufficiently large, this is less than, say, $1/(1000K)$ (where $K:=|\Clauses_k^L|$) and therefore we can take a union bound over all the $|\Clauses_k^L|$ patterns $\vecc$.
\end{proof}

And similarly, we develop an algorithm for bounded-degree instances:

\begin{proof}[Proof of \cref{cor:bounded-degree}]
    Assume for simplicity we are given as input $m$, the number of edges in $\Psi$. We can further assume $m \geq \Omega(n^{1-1/k})$ since otherwise we can store all constraints in the instance and calculate the $\mkand$ value exactly. Let $p = (CD/(m\epsilon^2))^{1/k}$ where $C=C(k)$ is a large constant to be chosen later. Consider the following algorithm: Before the stream, let $S \subseteq [n]$ be sampled by including every variable in $[n]$ with probability $p$ independently.\footnote{Two technical points, see \cite[\S3]{SSSV23-random-ordering} for details: (1) We can avoid assuming we have $m$ as input via standard tricks. We can ``guess'' a geometric progression of values for it and run the algorithm in parallel for these logarithmically many guesses. One guessed value will be within a constant factor of $m$, which we will use to produce the estimate; we can implement a space cutoff on every run of the algorithm to ensure we do not run out of space. (For the same reason, below we will only prove that the space bound holds with large probability.) (2) We can sample $S$ ``on-the-fly'' using a hash function, where $2k$-wise independence will suffice, instead of storing it up-front.} Now, during the stream, we store every constraint whose variables are all in $S$ in a set $E$, and we also track the bias of every variable in $S$. Finally, we estimate $\Snap$ via $\hat{M}(\vecc) = \frac1{p^k m} \hat{X}_\vecc$ where $\hat{X}_\vecc$ counts the number of constraints in $E$ with pattern $c$. (Note that we know the pattern of every constraint in $E$ because we know the bias of every variable in $S$.) We claim (1) that this sampling method runs within the space bound, and (2) that for every $\vecc \in \Clauses_k^L$, $|\hat{M}(\vecc) - \Snap(\vecc)| \leq \epsilon$.

    \paragraph{Space bound.} Let $n^*$ denote the number of nonisolated variables in $\Psi$, i.e., the number of variables occurring in at least one constraint. Observe that \[ \Exp[|S|] = p n^* = \left(\frac{CD}{\epsilon^2}\right)^{1/k} n^* m^{-1/k}. \] Now using $m \geq n^*$ and then $n^* \leq n$ we get $\Exp[|S|] \leq (CD/\epsilon^2)^{1/k} n^{1-1/k}$. Also, we have \[ \Exp[|E|] = p^k m = \frac{CD}{\epsilon^2}, \] which is constant as a function of $n$. Therefore using Markov's inequality, for sufficiently large $n$ we use $O(D^{1/k} n^{1-1/k} \log n/\epsilon^{2/k})$ space with probability $999/1000$.

    \paragraph{Correctness.} For each pattern $\vecc$, let $X_\vecc = \Snap(\vecc) \cdot m$ denote the true number of constraints in $\Psi$ with pattern $\vecc$; recall that $\hat{X}_\vecc$ counts the number of constraints in $E$ with pattern $\vecc$. By linearity of expectation, we have 
    \begin{equation}\label{eq:bounded-deg:hatX-exp}
        \Exp[\hat{X}_{\vecc}] = p^k X_\vecc.
    \end{equation}
    Therefore $\Exp[\hat{M}(\vecc)] = \Snap(\vecc)$ (where, recall, $\hat{M}(\vecc) = \frac1m p^{-k} \hat{X}_\vecc$ is our estimate for $\Snap(\vecc)$). So, it remains to prove concentration of $\hat{M}_\vecc$.

    Fix $\vecc \in \Clauses_k^L$. By definition of $\hat{M}_\vecc$ and $X_\vecc$, we want to prove that $|X_\vecc - \hat{X}_\vecc| \leq \epsilon p^k m$. Let $F := \{j \in [m] : \ptn(C_j) = \vecc\}$ be the set of all clauses with pattern $\vecc$. Then we can write $\hat{X}_{\vecc} = \sum_{j \in F} Y_j$ where $Y_j$ is the indicator for the event that all of $C_j$'s variables are sampled in $S$.
    
    Thus, $\Var[\hat{X}_{\vecc}] = \sum_{j,j' \in F} \Cov[Y_j,Y_{j'}]$. Hence we can bound $\Var[\hat{X}_{\vecc}]$ using that (a) $Y_j$ and $Y_{j'}$ are independent (and therefore $\Cov[Y_j,Y_{j'}]=0$) whenever $C_j$ and $C_{j'}$ do not share any variables; (b) by the maximum-degree assumption, any particular $C_j$ can share variables with at most $kD$ other constraints $C_{j'}$; and (c) even if $C_j$ and $C_{j'}$ do share a variable, we still have $\Cov[Y_j,Y_{j'}] \leq \Exp[Y_j Y_{j'}] \leq \Exp[Y_j]$ since $Y_j$ and $Y_{j'}$ are $\{0,1\}$-valued variables. Therefore,
    \begin{equation}\label{eq:bounded-deg:hatX-var}
        \Var[\hat{X}_{\vecc}] \leq (kD+1) \Exp[\hat{X}_\vecc].
    \end{equation}
    
    Finally, using Chebyshev's inequality and doing some manipulations, we have:

    \begin{align*}
        \Pr[|\hat{X}_\vecc - X_\vecc| \leq \epsilon p^k m] &\leq \frac{\Var[\hat{X}_\vecc]}{\epsilon^2 p^{2k} m^2} \tag{Chebyshev's inequality} \\
        &\leq \frac{(kD+1)\Exp[\hat{X}_\vecc]}{\epsilon^2 p^{2k} m^2} \tag{\cref{eq:bounded-deg:hatX-var}} \\
        &= \frac{(kD+1) X_\vecc}{\epsilon^2 p^k m^2} \tag{\cref{eq:bounded-deg:hatX-exp}} \\
        &\leq \frac{kD+1}{\epsilon^2 p^k m} \tag{$X_\vecc\leq m$} \\
        &= \frac{kD+1}{CD} \tag{def. of $p$}.
    \end{align*}
    
    For a sufficiently large choice of $C=C(k)$, again, we get the desired concentration (even union-bounding over $\vecc$).
\end{proof}

\ifnum\doubleblind=0
\section*{Acknowledgements}

I would like to thank Madhu Sudan and Santhoshini Velusamy for generous feedback and comments on the manuscript, and also Pravesh Kothari and Peter Manohar for helpful discussions.

This material is based upon work supported by the National Science Foundation Graduate Research Fellowship Program under Grant No. DGE2140739. Any opinions, findings, and conclusions or recommendations expressed in this material are those of the author and do not necessarily reflect the views of the National Science Foundation.
\fi

\printbibliography

\end{document}